\documentclass[11pt]{article}

\usepackage[utf8]{inputenc}
\usepackage[T1]{fontenc}
\usepackage{euscript}
\usepackage{amsmath,amsthm}
\usepackage{amssymb}
\usepackage{mathtools}
\usepackage{bm}
\usepackage{tikz}
\usetikzlibrary{decorations.pathreplacing}
\usepackage{url}
\usepackage[a4paper]{geometry}

\numberwithin{equation}{section}

\newcommand{\mB}{\mathcal{B}}

\newcommand{\mA}{\mathcal{A}}

\newcommand{\mX}{\mathcal{X}}
\newcommand{\mG}{\mathcal{G}}

\newcommand{\bZ}{\mathbb{Z}}
\newcommand{\bN}{\mathbb{N}}

\newcommand{\bF}{\mathbb{F}}

\newcommand{\Z}{\EuScript Z}

\newcommand{\N}{\EuScript N}
\newcommand{\C}{\EuScript C}
\newcommand{\X}{\EuScript X}
\newcommand{\cQ}{\mathcal Q}
\newcommand{\bfX}{\mathbf X}
\newcommand{\bfY}{\mathbf Y}
\newcommand{\ff}{\mathbf f}
\newcommand{\f}{\mathbb F}
\newcommand{\x}{\mathbf x}
\newcommand{\y}{\mathbf y}
\newcommand{\z}{\mathbf z}
\newcommand{\s}{\mathbf s}

\renewcommand{\u}{\mathbf u}
\renewcommand{\v}{\mathbf v}
\renewcommand{\c}{\mathbf c}
\newcommand{\error}{\mathbf e}

\newcommand{\restVA}{\raisebox{-0.5mm}{|}{\raisebox{-1mm}{\scriptsize V\!\! A}}}
\newcommand{\restEB}{\raisebox{-0.5mm}{|}{\raisebox{-1mm}{\scriptsize EB}}}
\newcommand{\restTa}{\raisebox{-0.5mm}{|}{\raisebox{-1mm}{\scriptsize Ta}}}

\renewcommand{\geq}{\geqslant}
\renewcommand{\leq}{\leqslant}
\newcommand{\rank}{\operatorname{rank}}
\newcommand{\Id}{\operatorname{Id}}

\def\Im#1{{\rm Im}\,#1}

\newtheorem{theorem}{Theorem}[section]
\newtheorem{proposition}[theorem]{Proposition}
\newtheorem{lemma}[theorem]{Lemma}
\newtheorem{corollary}[theorem]{Corollary}

\theoremstyle{Definition}
\newtheorem{definition}[theorem]{Definition}

\theoremstyle{plain}

\title{Decodable quantum LDPC codes beyond the $\sqrt{n}$ distance barrier using high
dimensional expanders}

\date{April 16, 2020}

\author{Shai Evra\thanks{Institute for Advanced Studies, Princeton, USA. {\tt shai.evra@gmail.com}} \and
Tali Kaufman\thanks{Department of Computer Science, Bar-Ilan
University, Israel. {\tt kaufmant@mit.edu} Research supported by ERC.} \and Gilles Z\'emor\thanks{Institut de Math\'ematiques
de Bordeaux, UMR 5251, France. {\tt zemor@math.u-bordeaux.fr}}
}

\begin{document}

\maketitle

\begin{abstract}
Constructing quantum LDPC codes with a minimum distance that grows faster than 
a square root of the length has been a major challenge of the field. With this challenge in mind, 
we investigate constructions that come from high-dimensional expanders, in
particular Ramanujan complexes. These naturally give rise to very unbalanced
quantum error correcting codes that have a large $X$-distance but a much smaller
$Z$-distance. However, together with a classical expander LDPC code and a
tensoring method that generalises a construction of Hastings and also the
Tillich-Zemor construction of quantum codes, we obtain quantum LDPC codes  whose
minimum distance exceeds the square root of the code length and whose dimension
comes close to a square root of the code length. When the ingredient is a
3-dimensional Ramanujan complex, we show that its 2-systole behaves like a
square of the log of the complex size, which results in an overall quantum code
of minimum distance $n^{1/2}\log n$, and sets a new record for quantum LDPC codes.
When we use a 2-dimensional Ramanujan complex, or the 2-skeleton of a
3-dimensional Ramanujan complex, we obtain a quantum LDPC code of minimum
distance $n^{1/2}\log^{1/2}n$. We then exploit the expansion properties of the
complex to devise the first polynomial time algorithm that decodes above the
square root barrier for quantum LDPC codes.
\end{abstract}

\clearpage
\section{Introduction}
A quantum CSS code \cite{CS96,Ste96} of length $n$ is defined by two binary matrices $H_X$ and $H_Z$,
each with $n$ columns, and such their row-spaces $W_X$ and $W_Z$ are orthogonal.
The matrices $H_X$ and $H_Z$ can be thought of as the parity-check matrices of
classical codes, $C_X=W_X^\perp$ and $C_Z=W_Z^\perp$ respectively.
The dimension of the quantum code is given by $n-\dim W_X - \dim W_Z$, equivalently
it is the dimension of either of the quotient spaces $C_X/W_Z$ or $C_Z/W_X$.
The Hamming distance $d_X$ (respectively $d_Z$) is defined as the smallest weight of a vector of
$C_X$ not in $W_Z$ (respectively $C_Z$ not in $W_X$). The minimum distance $d$
of the quantum code is defined as $d=\min(d_X,d_Z)$.
A quantum CSS code is said to be Low Density Parity Check (LDPC) if both
matrices $H_X$ and $H_Z$ have row and column weights bounded from above by a
constant. 

Quantum LDPC error correcting codes are the subject of a lot of ongoing research.
One reason is that quantum computers will need some form of quantum error
correction, and it is generally assumed that the relevant error correcting codes
will be of LDPC type because the associated quantum states can then be
constructed through local interaction between qubits. Other motivations come
from
quantum complexity theory: for example, the ``no low-energy trivial state'' 
conjecture \cite{HasNLTS}, generally thought of as a milestone towards a quantum
PCP theorem, involves quantum LDPC codes.

Constructing quantum LDPC codes with a minimum distance that grows with~$n$ has
been something of a challenge: one major difference with classical LDPC codes
is that choosing a sparse parity-check matrix at random gives with very high probability an
asymptotically good classical code, i.e. with dimension and minimum distance that scale as
linear functions of the blocklength $n$. For the very same reason, there are no
known random constructions of quantum LDPC codes, because choosing a matrix
$H_X$ at random will forbid the existence of a sparse matrix $H_Z$
in the dual space of the rowspace $W_X$ of
$H_X$. All known constructions of quantum LDPC codes are in contrast highly
structured. It is a wide open problem as to whether there exist families of
asymptotically good quantum LDPC codes. More specifically, known quantum LDPC
codes do not surpass a $\sqrt{n}$ barrier for the quantum
minimum distance. Families of quantum LDPC codes include the Kitaev code
\cite{Ki}, the
earliest and most studied LDPC construction, one version of which has parameters
$[[n,2,\sqrt{n}]]$, generalisations to surface codes \cite{BT,Z}, where qubits are
associated to the edges of a graph that tiles a surface: when the rate of these
codes is constant the minimum distance grows at best like $\log n$ \cite{De},
hypergraph product codes \cite{TZ} that have a constant rate and minimum
distance scaling like $\sqrt{n}$, the cubic codes of \cite{CDZ}, codes from
4-dimensional hyperbolic manifolds \cite{GL,LL} that have constant rate and minimum
distance $n^\alpha$ with $0.1\leq\alpha \leq 0.3$, iterated tensor power
constructions~\cite{AC}. There has been just one
construction, due to Freedman, Meyer and Luo \cite{FML}, that managed to break
through the square root barrier for the minimum distance, yielding a quantum
code of dimension $2$ and distance that scales like $n^{1/2}\log^{1/4}n$\footnote{The paper \cite{FML} advertises
$n^{1/2} \log^{1/2} n$ but this is a minor miscomputation.}.
A construction of Hastings \cite{Has16} has been conjectured to yield codes with minimum
distance close to linear in $n$, but does not provably break through the
$\sqrt{n}$ barrier. A construction of Bravyi and Hastings \cite{BH} does yield
asymptotically good quantum CSS codes, but at the expense of relaxing the LDPC
condition, namely the matrices $H_X$ and $H_Z$ have rows of Hamming weight $\sqrt{n}$.

It is arguably one of the most intriguing problems of the theory of quantum LDPC
codes, as to whether there exist codes whose minimum distance significantly
exceed the $\sqrt{n}$ barrier. In the present work we contribute to this
question by exhibiting codes that go beyond the Freedman et al. lower bound, and
set a new record for the minimum distance that scales as $n^{1/2}\log n$. The
dimension of these codes comes close to $\sqrt{n}$.
The way this is achieved is by calling upon some remarkable properties of {\em
Ramanujan complexes}. Ramanujan complexes are simplicial complexes that
generalise Ramanujan graphs and have higher-dimensional expansion properties. 
The $2$-dimensional Ramanujan LSV complexes of \cite{LSV2} can be thought of as a
graph every edge of which belongs to a fixed number of triangles. By associating
qubits to edges and using for $H_X$ and $H_Z$ the vertex-edge incidence matrix
and the triangle-edge incidence matrix, one defines a quantum LDPC code such
that $d_X=\log n$ and $d_Z=\Omega(n)$. Strictly speaking, this only yields
a minimum distance equal to $\log n$, however this code has the remarkable
property that $d_Xd_Z=\Omega(n\log n) \gg n$. A method of Hastings
\cite{Has17} allows one to make a new quantum code out of $d_Z/d_X$ copies of
the original one, yielding a code of length $nd_Z/d_X$, the same dimension as
the original dimension (in this case a constant), and minimum distance equal to $d_Z$. 
This already yields a code of length $n$ and of minimum distance
$\Omega(\sqrt{n\log n})$.
In the present paper we further investigate how LSV complexes can yield
good quantum error-correcting codes. We improve the dimension of the resulting
quantum LDPC code by replacing the Hastings construction with a more general
tensoring operation of complexes which will boost the code dimension to
something close to $\sqrt{n}$. 
This construction can be seen as generalisation of the construction of \cite{TZ}
where a component bipartite graph is replaced by a $2$-dimensional chain complex.
We will prove that starting from a
$3$-dimensional LSV complex, the tensoring construction yields the record
minimum distance $\Omega(n^{1/2}\log n)$. This involves obtaining a new systolic lower bound of the
form $\log^2 n$ for these complexes. We will also prove a systolic lower bound
of the form $\log^{k-1}n$ for $k$-dimensional LSV complexes, potentially
yielding quantum LDPC codes with minimum distance $\Omega(\sqrt{n\log^{k-1} n})$
for abitrary $k$, but the dimension of these codes is for now only conjecturally
non-zero.

Our main focus will then be to study in detail the decoding problem for codes that
come from $2$-dimensional LSV complexes and achieve minimum distance
$\Omega(\sqrt{n\log n})$. This involves using an auxiliary classical expander
code to reduce the decoding problem
to that of the unbalanced quantum code associated to the component simplicial
LSV complex. We then use the coboundary expansion properties of the LSV complex
to solve the remaining decoding problem. We also give an alternative decoding
procedure when the $2$-dimensional LSV complex is replaced by the $2$-skeleton
of a $3$-dimensional complex.

\section{Overview}
\subsection{CSS codes from simplicial complexes and from LSV complexes}
A quantum CSS code is defined by two binary matrices $H_X$ and $H_Z$ such that
$H_Z^TH_X=0$. If we call $X_0,X_1,X_2$ the sets of rows of $H_X$,
columns of $H_X$ (or of $H_Z$), rows of $H_Z$, then $H_Z^T$ and $H_X$ are the
matrices of two linear maps $\partial_2$ and $\partial_1$ 
\[
\f_2^{X_2}\xrightarrow{\partial_{2}} \f_2^{X_1}\xrightarrow{\partial_{1}}
\f_2^{X_0}
\]
such that $\partial_1\partial_2=0$. More generally,
a {\em chain complex} (of binary vector spaces) $\bfX=(X_0,X_1,\ldots , X_d)$ of
dimension $d$, describes a collection of vector spaces of the form $\f_2^{X_p}$
together with linear maps $\partial_p: \f_2^{X_p}\to\f_2^{X_{p-1}}$,
$p=1\ldots d$, such that $\partial_{p-1}\partial_p=0$ for $p=2\ldots d$. The
maps $\partial_p$ are called differential or boundary operators. We can
therefore extract a CSS code from any two consecutive differential
operators of a chain complex. This of course does not tell us very much about
which chain complexes are likely to give us interesting quantum codes,
but it is natural to focus on {\em simplicial} complexes.
A complex is simplicial when elements of $X_p$ describe
$(p+1)$-subsets $S$ of $X_0$ ($p$-simplices)
and the map $\partial_p$ takes the vector supported by $S$ to the vector
supported by the union of all $p$-subsets of $S$. The sets $X_0$ and $X_1$
describe therefore respectively the vertex and edge set of a graph, the set
$X_2$ describes a set of triangles in the graph, and so on.
When extracting the subcomplex $X_{p-1},X_p,X_{p+1}$ of a simplicial complex,
the rows of the matrix $H_X$ describing $\partial_p^T$ have weight $p+1$ and
$H_X$ is therefore LDPC for fixed $p$. The associated quantum code is therefore
LDPC when the complex is of {\em bounded degree}, meaning that every $p$-simplex is
incident to at most a bounded number of $(p+1)$-simplices.

The simplicial complexes that we shall use come from the recent theory of high dimensional expanders and Ramanujan complexes. 
Ramanujan complexes generalise Ramanujan graphs in sophisticated ways and
we will not define in all generality what 
they actually are, referring the interested reader to the excellent surveys  \cite{L1} and \cite{L2}.
We will however mention some of their remarkable properties, 
which are most relevant to us in the present work.
First of all, they can be explicitly constructed as clique complexes of Cayley or Schreier graphs associated to the finite groups $PGL_{d+1}(\mathbb{F}_{q^e})$ ($d$ the dimension of the complex, $q$ a prime power), as was done in \cite{LSV2}, and in fact we shall focus only on these constructions, henceforth called LSV complexes.
Secondly, their local structure displays excellent expansion properties,
notably if $\bfX=(V,E,T)$ is a $2$-dimensional LSV complex, and $L(v)$ its link
around the vertex $v\in V$, which is a graph whose vertex set is made up of the
neighbours of $v$ and such that vertices $u,w$ are connected in $L(v)$ if
$(u,v,w) \in T$, then for any $v$ the link $L(v)$ is isomorphic to the points versus lines incidence graph of a projective plane of order $q$. 
Third, as was shown in \cite{KKL}, for some of these $2$-dimensional LSV complexes their homology space $H_1=\ker\partial_1/\Im\partial_2$, whose dimension is exactly equal to the associated quantum code dimension, is non-zero, and similarly for $3$-dimension and second homology.
Fourth, following \cite{KKL} and \cite{EK}, these LSV complexes have cosystoles
which grow linearly in the size of the complex.
Fifth, we can construct LSV complexes (with non-trivial homology) whose
injectivity radius grows logarithmically in the size of the complex, and note
that the injectivity radius bounds from below the $1$-systoles.
From the above we can record the following result which follows essentially from the work of \cite{KKL}.

\begin{theorem}
\label{thm-dim-1} 
There exist a family of bounded degree $2$-dimensional LSV complexes $\bfX=(V,E
,T)$, such that the quantum code associated to it is non-zero and satisfies
\[
n = |E|,\quad k=\dim H^1(\bfX)>0, \quad d_X = S_1(\bfX)=\Omega(\log n),\quad d_Z= S^1(\bfX)=\Omega(n).
\]
 Where $S_1$ and $S^1$ are the $1$-systole and the $1$-cosystole
of the complex and are exactly equal to the minimum distances $d_X$ and $d_Z$ of
the associated quantum code.
\end{theorem}

The $1$-systolic distance (or $X$-distance) of the quantum code associated
to a $2$-dimensional simplicial complex is usually constrained by a $\log n$
upper bound. This is reminiscent of the girth of a regular graph (of degree $\geq 3$)
being bounded from above by $\log|V|$. However, when switching to the
$(X_{p-1},X_p,X_{p-1})$ subcomplex of a simplicial complex of dimension larger
than $2$, we may expect to achieve larger distances $d_X$.
$3$-dimensional LSV complexes $\bfX=(V,E,T,P)$, that on top of
triangles have 4-cliques (tetrahedra or Pyramids) 
in the underlying graph were also studied in \cite{KKL}, where it was shown that
their homology space $H_2=\ker\partial_2/\Im\partial_3$ is non-zero,
and that their $2$-cosystole behaves as $S^2(\bfX)=\Omega(|T|)$.
This translates into the quantum code associated to the subcomplex
$(E,T,Q)$ having non-zero dimension and $Z$-distance $d_Z=\Omega(n)$ where
$n=|T|$ is the code length. 
The question of the $2$-systole $S_2(\bfX)$ was left
unexplored. We prove a lower bound on the $2$-systole, and also on
higher-dimensional systoles: together 
 with results from \cite{KKL} and \cite{EK}, this gives:

\begin{theorem}
\label{thm:S^2}
\hspace{1cm}
\vspace{-3mm}
\begin{itemize}
\item[(i)] For $d$-dimensional LSV complexes $\bfX=(X_0,\ldots ,X_d)$, we have
\[
S_p(\bfX)=\Omega(\log^p|X_p|), \quad S^p(\bfX)=\Omega(|X_p|) \qquad \forall p=1,\ldots,d-1.
\]
\item[(ii)] If $d=3$, then there are LSV complexes for which $H^2(\bfX) \ne 0$, hence, its associated quantum code has parameters
\[
n=|X_3|,\quad k=\dim H^2(\bfX)>0,\quad d_X=\Omega(\log^2n), \quad d_Z=\Omega(n).
\]
\end{itemize}
\end{theorem}

To prove part $(i)$ of the Theorem, we invoke an injectivity radius argument,
together with arguments from building theory and
algebraic topology to claim that a non-trivial $p$-cycle must contain more
$p$-faces than in the intersection of an apartment and a ball of radius
$\log|X_p|$. An apartment is isomorphic to a tiling of $d$-dimensional Euclidean space
and Euclidean geometry arguments enable us to conclude.
 
Part $(i)$ of Theorem~\ref{thm:S^2} is quite general, and can in principle yield quantum LDPC codes with $d_X=\Omega(\log^j n)$ and $d_Z=\Omega(n)$ for $j>2$. 
However, the dimension of these quantum code candidates is only conjectured to be non-zero.

Next, we transform the quantum codes we have just discussed into
quantum LDPC codes with minimum distance larger than $\sqrt{n}$.

 
\subsection{Balancing distances $d_X$ and $d_Z$ of a quantum code}
We introduce the following construction of a quantum code.
it takes as input:
\begin{itemize}
\item A quantum code $\cQ=\cQ(\bfX)$ defined by two low-density parity-check matrices $H_X$
and $H_Z$. We can think of it as
coming from an abstract chain complex $\bfX=(X_0,X_1,X_2)$ where
$X_0$ and $X_2$ index the set of rows of $H_X$ and the set of rows of 
$H_Z$ respectively, and $X_1$ indexes both the set of columns of $H_X$ and the
set of of columns of $H_Z$. The matrices $H_X$ and $H_Z$ define incidence
relations between elements of $X_0$ and $X_1$ and between elements of $X_1$ and
$X_2$. We denote by $d_X(\cQ)$ and
$d_Z(\cQ)$ its $X$ and $Z$-distances.
\item A classical LDPC code $C=C(\bfY)$ 
 defined by a low-density parity-check matrix~$H$.
We can think of it as coming from a $1$-dimensional chain complex $\bfY=(A,B)$,
which just means that we index the columns of $H$ by a set $A$ and its rows by a
set $B$. The matrix $H$ defines an incidence relation between elements of $A$
and $B$. It is important that the matrix $H$ has no redundant rows, i.e.
$\rank(H)=|B|$. We denote by $d(C)$ its
minimum distance.
\end{itemize}
The construction outputs a new quantum code $\cQ(\X)$ by defining a
$2$-dimensional chain complex $\X=(\X_0,\X_1,\X_2)$, where
\begin{align*}
\X_0&= (X_0\times A)\cup (X_1\times B)\\
\X_1&= (X_1\times A)\cup (X_2\times B)\\
\X_2&= X_2\times A.
\end{align*}
To define incidence between elements of $\X_0$ and elements of $\X_1$,
we declare any $(x_1,a)\in X_1\times A$ to be incident to $(x_0,a)$ for all
$x_0\in X_0$ incident to $x_1$ in $\bfX$, and to be incident to
$(x_1,b)$ for all $b\in B$ incident to $a$ in $\bfY$. We also
declare any $(x_2,b)\in X_2\times B$ to be incident to $(x_1,b)$ for
all $x_1\in X_1$ incident to $x_2$ in $\bfX$.
To define incidence  between elements of $\X_2$ and elements of $\X_1$, we declare $(x_2,a)\in X_2\times A$ to be incident 
to $(x_1,a)$ for all $x_1\in X_1$ incident to $x_2$ in $\bfX$, and to be incident
to $(x_2,b)$ for all $b\in B$ incident to $a$ in $\bfY$. The factor graph
representation of the quantum is depicted on Figure~\ref{fig:product-overview}.
 
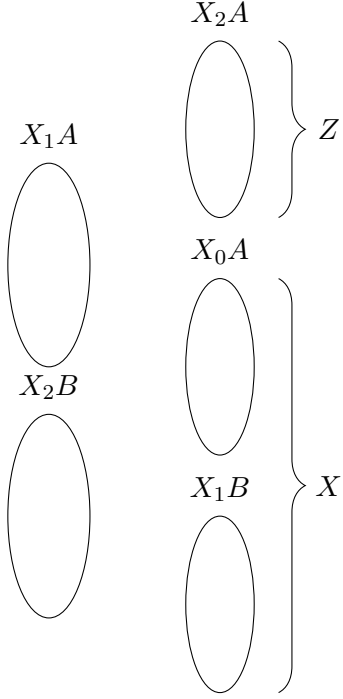
\begin{figure}[h]
\begin{center}
\begin{tikzpicture}[scale=0.9]
\draw (0,0) ellipse (0.6cm and 1.5cm);
\node (EA) at (0,1.9) {$X_1A$};
\draw (0,-3.7) ellipse (0.6cm and 1.5cm);
\node (TB) at (0,-1.8) {$X_2B$};
\draw (2.5,2) ellipse (0.5cm and 1.3cm);
\node (TA) at (2.5,3.7) {$X_2A$};
\draw [decorate,decoration={brace,amplitude=10pt,mirror,raise=4pt},yshift=0pt]
(3.2,0.7) -- (3.2,3.3) node [black,midway,xshift=0.8cm] {$Z$};
\draw (2.5,-1.5) ellipse (0.5cm and 1.3cm);
\node (VA) at (2.5,0.2) {$X_0A$};
\draw (2.5,-5) ellipse (0.5cm and 1.3cm);
\node (EB) at (2.5,-3.3) {$X_1B$};
\draw [decorate,decoration={brace,amplitude=10pt,mirror,raise=4pt},yshift=0pt]
(3.2,-6.3) -- (3.2,-0.2) node [black,midway,xshift=0.8cm] {$X$};
\end{tikzpicture}
\end{center}
\caption{The factor graph structure of the quantum code $\cQ(\X)$. The code length is
$N=|X_1||A|+|X_2||B|$.}
\label{fig:product-overview}
\end{figure}

From the definitions, we have that:
If $w_X^R,w_Z^R,w_Z^C$ are upper bounds respectively on the row weights of $H_X$,
the row weights of $H_Z$ and the column weights of $H_Z$, and if
$w^R,w^C$ are upper bounds respectively on the row weights and the column
weights of $H$, then $\cQ(\X)$ is LDPC with its $Z$-row weights $W_Z$ and $X$-row
weights $W_X$ being bounded
from above by 
\begin{align*}
W_Z&\leq w_Z^R +w^C\\
W_X&\leq \max(w_X^R,w^R+w_Z^C).
\end{align*}
Similar relations hold for $Z$ and $X$ column weights of the new code.

We prove:
\begin{theorem}
\label{thm:product-overview}
The resulting quantum LDPC code $\cQ(\X)$ has length, dimension, $X$-distance and
$Z$-distance equal to, respectively:
\begin{align*}
 N&=|X_1||A|+|X_2||B|\\
 K&=\dim\cQ\dim C\\
 D_X&=d_X(\cQ)d(C)\\
 D_Z&=d_Z(\cQ).
\end{align*}
\end{theorem}

Theorem~\ref{thm:product-overview} generalises a construction of Hastings
\cite{Has17} that corresponds to the special case when the complex $\bfY=(A,B)$
is simplicial and when the underlying graph is a simple path. In coding theory
terms, it is the case when the classical code $C$ has dimension~$1$ and is the
repetition code. The construction \cite{TZ} of quantum LDPC codes is also a
special instance of it, corresponding to the case when $X_0=\emptyset$, meaning
that there is no parity-check matrix $H_X$ and that $d_X=1$: in other words
$\cQ(\bfX)$ is reduced to a classical LDPC code.

Theorem~\ref{thm:product-overview} tells us therefore that if we have a quantum
code $\cQ$ such that $d_Z\gg d_X$, we can apply to it the $\X$-construction
using
a classical code $C$ with minimum distance $d\approx d_Z/d_X$,
and obtain a new quantum code with $D_X\approx D_Z$.

Specifically, starting with the ``Ramanujan'' quantum codes of the previous
section, and using for the classical code $C$ an asymptotically good LDPC code,
i.e. with dimension and minimum distance that are linear in its blocklenth,
(such a code being known to exist, either through the random methods that go back to
Gallager, or through the expander code construction of Sipser and Spielman
\cite{SS}), we obtain:

\begin{corollary}\label{cor:record}
\hspace{1cm}
\vspace{-3mm}
\begin{itemize}
\item[(i)] 
$2$-dimensional Ramanujan complexes yield a family of quantum LDPC 
codes of length $N$, dimension $K$, and minimum distance $D$, with
\[
K=\Omega\left(\sqrt{\frac{N}{\log N}}\right),\qquad
D=\Omega(\sqrt{N\log N}).
\]
\item[(ii)]
$3$-dimensional Ramanujan complexes yield a family of quantum LDPC codes of length $N$, dimension $K$, and minimum distance $D$, with
\[
K=\Omega\left(\frac{\sqrt{N}}{\log N}\right),\qquad
D=\Omega(N^{1/2}\log N).
\]
\end{itemize}
\end{corollary}

\subsection{The decoding problem}
We now address the decoding problem for the codes of Corollary~\ref{cor:record}.
The objective is to correct any pattern of errors up to a constant fraction of the
distance. We do not know how to do it for the codes of $(ii)$ in
Corollary~\ref{cor:record}, but we achieve it for the 
codes of $(i)$.

A CSS code of length $n$ comes with two syndrome maps, namely
\begin{alignat*}{6}
\sigma_X~:\ &\f_2^{X_1}&&\to\; &&\f_2^{X_0}\qquad\qquad \sigma_Z~:\;
&&\f_2^{X_1}&&\to\; &&\f_2^{X_2}\\
&\x && \mapsto && H_X\x^T
          && \x &&\mapsto &&H_Z\x^T
\end{alignat*}
where $(X_0,X_1,X_2)$ is the associated complex. The syndrome maps
$\sigma_X$ and $\sigma_Z$ are also the boundary and coboundary operators of the complex
Let $\error=(\error_X,\error_Z)$ be a couple of vectors of $\f_2^n$, each of weight at most
$t$.
The {\em decoding problem} for a CSS code is, given $\sigma_X(\error_X)$ and
$\sigma_Z(\error_Z)$, to recover an {\em equivalent} version of $\error$, namely
a vector $\error'=(\error_X',\error_Z')$ such that $\error_X+\error_X'\in
C_Z^\perp$ and $\error_Z+\error_Z'\in C_X^\perp$.

In \cite{LTZ15}, which is a particular instance of the product $\X$ construction
just discussed, a decoding algorithm was devised that relied on expansion of the
two underlying complexes (simply graphs in this case). In the present case, we cannot hope for such
an approach to be completely transposed because though the underlying component
simplicial complex~$\bfX$ exhibits remarkable coboundary expansion, it does not
have boundary expansion. However, the product quantum code $\cQ(\X)$
 has a
property that the codes of \cite{LTZ15} do not have: if the component quantum code
$\cQ(\bfX)$ can be decoded both from $X$-errors and from $Z$-errors,
(which does not happen for the codes of \cite{LTZ15} since the
component ``quantum code'' is really a classical code and corrects
zero $X$-errors), and if the classical code $C$ also decodes a
linear fraction of errors, then the overall code $\cQ(\X)$ can be decoded from
both $X$-errors and $Z$-errors.

So we need a component quantum code that we can correct from both types of
errors. For $2$-dimensional simplicial complexes, boundary decoding (from
$X$-errors) comes
naturally, because it can be handled by complete decoding of cycle codes of
graphs, which is known to be achievable in polynomial time through minimal
weight matching in graphs. Focusing on $2$-dimensional simplicial complexes
$\bfX$, we
prove the reduction:

\begin{theorem}
\label{thm:reduction}
\hspace{1cm}
\vspace{-3mm}
\begin{itemize}
\item[(i)]
Suppose the classical LDPC code $C$ comes with a polynomial-time decoding
algorithm that corrects any pattern of less than $\alpha|A|$ errors.
Then there is a polynomial time algorithm for $\cQ(\X)$
that corrects all $X$-errors of weight
smaller than $\alpha|A|d_X/2$
where $d_X$ is the $1$-systole or $X$-distance for the component code
$\cQ(\bfX)$.
\item[(ii)] 
Suppose there is a polynomial time decoding algorithm for the component quantum code $\cQ(\bfX)$ that corrects any pattern of
$Z$-errors of weight smaller than $w$. Then there exists a
polynomial time algorithm for $\cQ(\X)$ that corrects any pattern of $Z$-errors
of weight smaller than $w$. 
\end{itemize}
\end{theorem}

To have a solution to the decoding problem for the product code $\cQ(\X)$, it
remains to find a decoding algorithm for coboundary decoding (from $Z$-errors)
of the component quantum code $\cQ(\bfX)$. 
We achieve this using coboundary expansion of LSV complexes in two different
ways. For a $2$-dimensional complex $\bfX=(V,E,T)$
the algorithm takes the following form:

\paragraph{Decoding algorithm:}\hspace{1cm}

{\it Input:} the coboundary or $Z$-syndrome $\ff_0=\sigma_Z(\error)$ for a $Z$-error
$\error$.

{\it Procedure:} for $k\geq 1$, look for a vertex $v$, and a vector
$\y_k\in\f_2^E$, whose support is entirely in the edge-neighbourhood of $v$,
such that $|\sigma_Z(\y_k)+\ff_{k-1}|<|\ff_{k-1}|.$ Set
$\ff_k=\ff_{k-1}+\sigma_Z(\y_k)$. Repeat until $\ff_k=0$ and output
$\error'=\y_1+\y_2+\cdots\y_k$.

We prove:

\begin{theorem}\label{thm:local-decoding-2d}
There exist constants $c,q_0$, such that when $\bfX$ is any $2$-dimensional LSV
complex of local parameter $q>q_0$, any $Z$-error vector $\error\in\f_2^E$ of weight
$|\error|\leq c|E|$ is always correctly decoded by the decoding algorithm.
\end{theorem}

The algorithm of Theorem~\ref{thm:local-decoding-2d} is linear-time in the
code length $|E|$, but with a constant that is exponential in the local
parameter $q$. We can remove the constant when we replace the
$2$-dimensional LSV complex by the $2$-skeleton $\bfX=(V,E,T)$
 of a $3$-dimensional LSV complex
$(V,E,T,P)$. Note that this differs from taking its $(E,T,P)$ subcomplex which
is used to create the codes of Corollary~\ref{cor:record} $(ii)$.
The associated product quantum code $\cQ(\X)$ will again have parameters 
equivalent to those of Corollary~\ref{cor:record} $(i)$, though with looser
constants. But in return we prove:
\begin{theorem}\label{thm:local-decoding-3d}
For $q$ fixed and large enough, there exists a constant 
$c'$, such that when $\bfX$ is the $2$-skeleton of any 
$3$-dimensional LSV
complex of local parameter $q$, any $Z$-error vector $\error\in\f_2^E$ of weight
$|\error|\leq c|E|$ is always correctly decoded by the decoding algorithm using
local vectors $\y_k\in\f_2^E$ of weight $1$.
\end{theorem}
Note that the constant $c'$ in Theorem~\ref{thm:local-decoding-3d}
depends on $q$, as opposed to the constant $c$ in
Theorem~\ref{thm:local-decoding-2d} which is universal.

\subsection{Comments and open questions}
\begin{itemize}
\item
The ``expander'' quantum LDPC codes of \cite{LTZ15} can be seen as having been
constructed through a co-complex tensoring operation with two $1$-dimensional chain complexes, i.e. two bipartite graphs that are taken to be expanding graphs. 
The present quantum codes are obtained by replacing one of the components by a
coboundary expanding simplicial complex. In both cases expansion is crucial to
decoding, even though we rely on different decoding strategies. We have focused
on using Ramanujan LSV complexes, but good quantum codes are also liable to come from
more general families of higher-dimensional expanders. What is needed is a
simplicial complex with sufficiently good local expansion in its links: from
this global expansion properties can be derived \cite{O} and the required coboundary
expansion follows. To obtain a quantum
code of non-zero dimension one furthermore needs non-zero homology, and to
obtain systolic bounds one requires the existence of a covering complex with
zero homology and a growing injectivity radius.

\item Reasonable values can be given for the constants in Theorem~\ref{thm-dim-1}: for large enough $q$, we have that $d_Z=S^1(\bfX)$ is at
least a quantity
arbitrarily close to $n/4$ and $d_X\geq \frac{1}{32}\log_qn -3$.
The number of $Z$-errors correctable by the algorithm in 
Theorem~\ref{thm:local-decoding-2d} can be made arbitrarily
close to a $1/144$ fraction of the distance $d_Z$.
The constants in Theorem~\ref{thm:S^2} and Theorem~\ref{thm:local-decoding-3d}
are much looser.
\item
It is possible to show that the logarithmic behaviour of the systolic distance $d_X$ in Theorem~\ref{thm-dim-1} cannot be improved.
However it is very much open as to whether the $\log^2n$ lower bound 
on $d_X$ in Theorem~\ref{thm:S^2} is best possible or not. Any improvement would
of course mean an improvement over the minimum distance of the quantum code of
Corollary~\ref{cor:record} $(ii)$.
\item
Coboundary decoding of $2$-dimensional LSV complexes and of $2$-skeletons of
$3$-dimensional complexes $\bfX$ is linear in their length $|E|$, which translates into
linear-time decoding from $Z$-errors for the quantum code $\cQ(\X)$.
However for $X$-errors we need to rely on complete decoding of 
the cycle code associated to the $1$-skeleton of the complex $\bfX$
which is not linear in $|E|$ and we do not obtain a linear-time decoding
algorithm for $\cQ(\X)$. It would of course be interesting to find an
alternative strategy that would result in linear-time decoding from $X$-errors.
\item
The distance record-breaking higher-dimensional code of
Corollary~\ref{cor:record} $(ii)$ can be decoded from $Z$-errors by the same
strategy as that of Theorem~\ref{thm:local-decoding-3d}
and Section~\ref{sec:coboundary-decoding-3d}, if we replace
the component $(E,T,P)$ $2$-complex of a $3$-dimensional LSV complex
by the $(E,T,P)$ component of a {\em $4$-dimensional} LSV complex.
How to decode the corresponding quantum code from $X$-errors has eluded us
however.
\item
We have focused on decoding worst-case errors. It would be interesting to
address the decoding problem for random errors. The $X$-error decoding algorithm
for the quantum code $\cQ(\X)$ is easily seen to work just as well for random
errors up to a positive fraction of the code length. However, our $Z$-error
correcting strategy fails for linear weight random errors and a different
approach needs to be devised.
\end{itemize}

\subsection*{Outline of the manuscript}  
In Section~\ref{sec-per} we review the relevant algebraic and coding-theoretic
background that we need.
In Section~\ref{sec-prod} we give the details of the construction of the product
complex $\X$ and prove Theorem~\ref{thm:product-overview}.
Section~\ref{sec-Ram} is devoted to Ramanujan complexes and to the proof of
Theorems~\ref{thm-dim-1} and \ref{thm:S^2}.
Section~\ref{sec-dec-prod} describes how to decode the product 
complex $\cQ(\X)$
and proves Theorem~\ref{thm:reduction}. In Section~\ref{sec-dec-Ram} 
we give a refined analysis of coboundary expansion of $2$-dimensional LSV
complexes and prove Theorem~\ref{thm:local-decoding-2d}. Finally in
Section~\ref{sec:coboundary-decoding-3d} we prove
Theorem~\ref{thm:local-decoding-3d} for $2$-skeletons of $3$-dimensional LSV
complexes.

\subsection*{Acknowledgments}
We wish to thank Anthony Leverrier for pointing out to us the relevance of reference~\cite{Has17} and setting us on the path leading to this paper. We are grateful to the Israeli
Institute for Advanced Studies and the Simons Institute for the theory of
Computing for hosting programmes on high-dimensional expanders that fostered
this work.

\clearpage
\section{Preliminaries} \label{sec-per}

\subsection{Chain complexes}
Chain complexes are a useful formalism for studying quantum CSS codes~: it has
been developed and applied in particular in \cite{BH,AC}.

A $d$-dimensional {\it chain complex} $\bfX$, is 
a sequence running on $p=0,1,\ldots,d$ of binary vector spaces, called the $p$-chain spaces, 
$C_p(\bfX)$. Each chain space $C_p$ comes with a distinguished basis $X_p$, called
the set of {\it $p$-faces} of $\bfX$, so that we have the identification
$C_p=\f_2^{X_p}$.
The chain spaces come together 
with a family of linear maps between them, called the {\it
$p$-boundary} operators, denoted 
\begin{equation}
\partial_{p}\,:\,C_{p}(\bfX)\rightarrow C_{p-1}(\bfX),\nonumber
\end{equation}
for $p=0..d+1$, with the convention $\partial_0~:C_0\to 0$ and
$\partial_{p+1}~: 0\to C_p$, and
such that the boundary of the boundary is identically zero:
\begin{equation}
\partial_{p}\circ \partial_{p+1}=0.\nonumber
\end{equation}
One defines the spaces of {\it $p$-cycles} and {\em $p$-boundaries},
\begin{equation}
B_{p}(\bfX):=\mbox{im}\partial_{p+1}\subseteq\ker \partial_{p}=:Z_{p}(\bfX),\nonumber
\end{equation}
and the quotient, called the {\em $q$-homology space} of $\bfX$, is 
\begin{equation}
H_{p}(\bfX):=Z_{p}(\bfX)/B_{p}(\bfX).\nonumber
\end{equation}
The adjoint operator of the $(p+1)$-boundary operator, called the {\em
$p$-coboundary operator}, is denoted
\begin{equation}
\delta_p=\partial_{p+1}^{*}\,:\,C_{p}^*(\bfX)\rightarrow C_{p+1}^*(\bfX)\nonumber
\end{equation}
Vectors of $C_p^*$ are called {\em cochains} but the space $C_p^*$ can be
identified with $C_p$.
Note that one also has that the coboundary of a coboundary is identically zero,
\begin{equation}
\delta_{p}\circ \delta_{p-1}=0.
\end{equation}
Similarly, one defines the spaces of  {\em $p$-cocycles} and {\em $p$-coboundaries}, 
\begin{equation}
B^{p}(\bfX):=\mbox{im}\delta_{p-1}\subseteq\ker \delta_{p}=:Z^{p}(\bfX)\nonumber
\end{equation}
and the {\em $p$-cohomology space} $H^1(\bfX)=Z^1(\bfX)/B^1(\bfX)$ of $\bfX$.\nonumber
Note that the following holds 
\begin{equation}
B^p(\bfX)^\perp  = Z_p(\bfX) \qquad \mbox{and} \qquad B_p(\bfX)^\perp = Z^p(\bfX).\nonumber
\end{equation}

Define the $p$-systole to be
\begin{equation}
S_{p}(\bfX)=\min\{|v|\,:\,v\in Z_{p}(\bfX)\setminus B_{p}(\bfX)\},\nonumber
\end{equation}
and the $q$-cosystole to be 
\begin{equation}
S^{p}(\bfX)=\min\{|v|\,:\,v\in Z^{p}(\bfX)\setminus B^{p}(\bfX)\}.\nonumber
\end{equation}

If $\bfX$ is a complex, it will be useful to denote by $\bfX^*$ its {\em co-complex},
defined by $\bfX^*_i=X_{p-i}$, and with $\delta_i$ as the boundary map from
$\bfX^*_i$ to $\bfX^*_{i-1}$.

If we are dealing with a family of $d$-dimensional complexes rather than an single complex 
$\bfX$,
then, following a standard abuse of terminology, we will say that the complex
$\bfX$ (as an abreviation for the family of complexes that it represents) has {\em
bounded degree} if for any $v\in X_{p}$, the weight of $\partial_{p}(v)$ and
$\delta_{p}(v)$ is bounded from above by a constant.

\paragraph{Simplicial Complexes.}
The $d$-dimensional complex $\bfX$ is said to be {\em simplicial}, if for every
$v\in X_p$, $p=1,\ldots d$, the weight of its boundary $\partial_p(v)$ is equal
to $p+1$.
This means in particular that the boundary map from $X_1$ to $X_0$ defines a
graph structure on $X_0$, with edge set $X_1$, every element $e$ of $X_1$
connecting the two vertices defined by $\partial_1(e)$. Similarly, the set $X_2$
defines a set of triangles in the graph, and more generally $X_p$ defines a set
of $(p+1)$ simplices. Conversely, given a graph we may define its {\em clique
complex} by defining $X_p$ to be the set of $(p+1)$-cliques in the graph.

A $1$-dimensional simplicial complex is therefore just a graph. In the same way
that we often write $(V,E)$ to describe a graph, with the incidence relation
between $V$ and $E$ being implicit, we will often allow ourselves to write
$\bfX=(X_0,X_1,X_2)$ for a $2$-dimensional complex, with the boundary operators
being implicit.

\subsection{CSS codes}\label{ssec:CSS}
\subsubsection{Classical codes}\label{sssec:classical}
Recall that a classical binary linear code $C$ of length $n$ is defined by a parity-check matrix
$H$ with $n$ columns as the set of vectors $\x\in\f_2^n$ such that $H\x^T=0$. 
It has dimension $n-\rank(H)$ and its minimum distance is defined as the
smallest weight of a non-zero codeword. The decoding problem is, given a vector
$\y\in\f_2^n$, to find the closest codeword for the Hamming distance. An
equivalent version is to be given the {\em syndrome} of $\y$, i.e. the quantity
$s=\sigma(\y)=H\y^T$, and to find the smallest weight vector $\error$ such that
$\sigma(\error)=s$. 
An infinite family of codes is said to be {\em asymptotically good}, if both the
code dimension and the code minimum distance are bounded from below by a
constant times the length $n$.

A family of codes is said to be LDPC (Low Density Parity
Check) if each code of the family has a parity-check matrix whose row-weight and
column weight are both bounded from above by a constant. It has been known since
Gallager first studied them that LDPC codes are asymptotically good. 

 An LDPC code is often
described by its {\em factor graph} (also called a Tanner graph) which is a
bipartite graph $(A,B)$ where $A$ (resp. $B$) is identified with the set of columns
(resp. rows) of $H$, and $H$ is a matrix representation of the $(A,B)$ incidence
structure. It was shown by Sipser and Spielman \cite{SS} that LDPC codes whose factor
graph is sufficiently expanding are asymptotically good and that this allows one
to construct asymptotically good LDPC codes, as opposed to showing they exist.

With the terminology of chain complexes, a code $C$ defined by the
matrix $H$ is a $1$-dimensional complex $\bfX$, with $X_0=B$, $X_1=A$, and with
boundary map the syndrome function given by $H$. The actual code $C$ as a
subspace of $\f_2^n=\f_2^A$ is the homology space $H_1$ and its minimum distance
is the $1$-systole $S_1(\bfX)$.

\paragraph{Cycle codes of graphs.} 
When the $1$-dimensional complex that defines $C$ is simplicial, we obtain a
{\em cycle code} of a graph on vertex set $X_0$. Cycle codes of graphs are often
not considered in the LDPC literature, in part because there are much better
codes: in particular if the code has dimension linear in $n$ a cycle code cannot
have minimum distance larger than $\log n$ (the behaviour of the girth of the
graph). However, cycle codes have an interesting property that will be useful to
us in the quantum coding context: they can be decoded completely in polynomial
time. This means that there is an efficient algorithm that,
given any vector $s$ belonging the syndrome (boundary)
space, finds a smallest weight vector (set of edges) that maps to $s$ by the
syndrome function (boundary map) \cite{NH}.

\subsubsection{Quantum codes}\label{sssec:quantum}
A CSS code of length $n$ is defined by two classical codes $C_X$ and $C_Z$,
associated with two parity-check matrices $H_X$ and
$H_Z$. The two codes $C_X$ and $C_Z$ define a CSS code if the condition
$C_X\supset C_Z^\perp$ (equivalently $C_Z\supset C_X^\perp)$ is satisfied.
In other words the rowspaces of $H_X$ and $H_Z$ should be orthogonal.
The dimension of the quantum code is defined by $\dim C_X/C_Z^\perp = \dim
C_Z/C_X^\perp$. We define the $X$-minimum distance $d_X$ as the smallest weight
of a vector of $C_X$ not in $C_Z^\perp$ and the $Z$-minimum distance as the
smallest weight of a vector of $C_Z$ not in $C_X^\perp$. The minimum distance
of the quantum code is defined as $d=\min(d_X,d_Z)$.

Let $X_0$ be the set of rows of $H_X$, $X_1$ the set of columns of either $H_X$
or $H_Z$, and $X_2$ the set of rows of $H_Z$. With the convention that vectors
are columns, we define the map $\partial_1~: \f_2^{X_1} \to \f_2^{X_0}$ by
left multiplication by $H_X$, and $\partial_2~: \f_2^{X_2} \to \f_2^{X_1}$ by
left multiplication by $H_Z^T$, and obtain a $2$-dimensional chain complex.
The codes $C_X$ and $C_Z$ are the cycle and cocycle spaces $Z_1$ and $Z^1$
respectively, the rows spaces of $H_X$ and $H_Z$ are the boundary and coboundary
spaces $B_1$ and $B^1$.
The homology and cohomology subspaces $H_1$ and $H^1$ are the subspaces 
$C_X/C_Z^\perp$ and $C_Z/C_X^\perp$ and the $X$ and $Z$-minimum distances are
the $1$-systolic and $1$-cosystolic constants respectively.
Conversely, any $2$-dimensional chain complex gives (actually is) a quantum CSS
code. 
Furthermore, any $d$-dimensional complex gives rise to a CSS code by
extracting from it the chain spaces $C_{p-1},C_p,C_{p+1}$ and the associated
boundary maps $\partial_p$ and $\partial_{p+1}$ that in matrix form yield $H_Z$
and $H_X$. The length of the code is $n=|X_p|$, its dimension is $\dim H_p(X)$,
and its $X$ and $Z$-distances are $S_p(X)$ and $S^p(X)$ respectively.

A quantum CSS code comes with two syndrome functions $\sigma_X$ and $\sigma_Z$
defined, for $\x\in\f_2^n$, as the maps $\x\mapsto H_X\x^T$ and $\x\mapsto
H_Z\x^T$ respectively. They may also be viewed as boundary and coboundary
operators respectively.
Let $\error=(\error_X,\error_Z)$ be a couple of vectors of $\f_2^n$, each of weight at most
$t$.
The {\em decoding problem} for a quantum CSS code is, given $\sigma_X(\error_X)$ and
$\sigma_Z(\error_Z)$, to recover an {\em equivalent} version of $\error$, namely
a vector $\error'=(\error_X',\error_Z')$ such that $\error_X+\error_X'\in
C_Z^\perp$ and $\error_Z+\error_Z'\in C_X^\perp$.

Finally, a quantum CSS code is said to be LDPC if it can be defined by matrices
$H_X$ and $H_Z$ that are low density, i.e. whose row and column weights are
bounded from above by a constant. As usual, when we talk about ``an LDPC code'',
we really mean a family of LDPC codes of growing lengths. As in the classical
case, it is convenient to represent the code by a factor graph, as in
Figure~\ref{fig:factor}.

\begin{figure}
\begin{center}
\begin{tikzpicture}
\draw (0,0) ellipse (0.8cm and 2cm);
\node (E) at (0,2.4) {$X_1$};
\draw (2.5,2) ellipse (0.6cm and 1.5cm);
\node (T) at (2.5,3.8) {$X_2$};
\draw [decorate,decoration={brace,amplitude=10pt,mirror,raise=4pt},yshift=0pt]
(3.2,0.5) -- (3.2,3.5) node [black,midway,xshift=0.8cm] {$Z$};
\draw (2.5,-2) ellipse (0.6cm and 1.5cm);
\draw [decorate,decoration={brace,amplitude=10pt,mirror,raise=4pt},yshift=0pt]
(3.2,-3.5) -- (3.2,-0.5) node [black,midway,xshift=0.8cm] {$X$};
\node (V) at (2.5,-0.2) {$X_0$};
\end{tikzpicture}
\end{center}
\caption{The factor graph of a CSS code~: ``variables'' are on the left. Two set of checks, the $X$-checks and
the $Z$-checks define the $\sigma_X$ and $\sigma_Z$ syndrome functions.}
\label{fig:factor}
\end{figure}
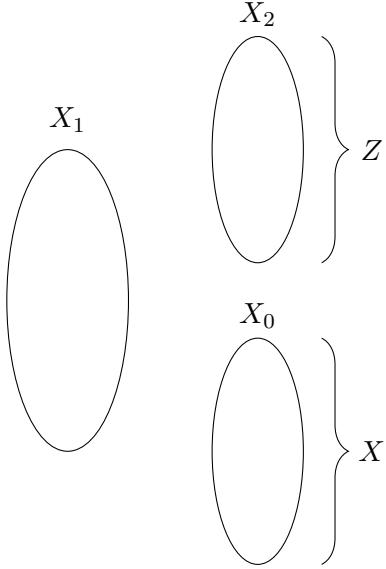

\section{Product complexes} \label{sec-prod}
The simplicial complexes that we will describe in the next section
yield quantum codes with very unbalanced distances $d_X$ and $d_Z$.
We describe how taking an appropriate product 
with a classical LDPC code yields a new quantum code with balanced $X$ and
$Z$-distances.

We first recall the following tensoring procedure (see e.g. \cite{AC}), which takes two complexes
$\bfX$, $\bfY$ and outputs a new complex $\bfX\otimes \bfY$, where 
\begin{equation}
(\bfX\otimes \bfY)_{p}=\bigsqcup_{i=0}^{p}X_{i}\times Y_{p-i}\nonumber
\end{equation}
hence the chain spaces are direct sums of tensor products
\begin{equation}
C_{p}(\bfX\otimes \bfY)=\oplus_{i=0}^{p}C_{i}(\bfX)\otimes C_{p-i}\bf(Y)\nonumber
\end{equation}
and the boundary maps acts on a tensor element $v\otimes u\in C_{i}(\bfX)\otimes
C_{p-i}(\bfY)$, as follows 
\begin{equation}
\partial_{p}^{\bfX\otimes \bfY}(v\otimes u)=\partial_{i}^{\bfX}(v)\otimes
u+v\otimes d_{p-i}^{\bfY}(u)\nonumber
\end{equation}
and $\partial_{p}^{\bfX\otimes \bfY}$ is extended linearly. 

It is important to note that if both $\bfX$ and $\bfY$ are bounded degree, then
so is $\bfX\otimes \bfY$. 

The homology of the product complex is described simply by
the K\"unneth formula (see e.g. \cite[Section~3.B]{Hat}), 
\begin{equation}
H_{p}(\bfX\otimes \bfY)\cong\oplus_{i=0}^{p}H_{i}(\bfX)\otimes H_{p-i}(\bfY).\nonumber
\end{equation}

which gives in particular:
\begin{equation}
|(\bfX\otimes \bfY)_{p}|=\sum_{i=0}^{p}|\bfX_{i}||\bfY_{q-i}|,\qquad\dim
H_{p}(\bfX\otimes
\bfY)=\sum_{i=0}^{p}\dim H_{i}(\bfX)\dim H_{p-i}(\bfY).\nonumber
\end{equation}

We now turn our attention to the specific case when $\bfX$ is a $2$-dimensional
$(d=2)$ complex and $\bfY$ is $1$-dimensional. As described in
Section~\ref{ssec:CSS}, $\bfX$ can be viewed as a quantum CSS code
of length $|X_1|$ whose $H_X$ and $H_Z$ matrices describe the
$(X_0,X_1)$ and $(X_2,X_1)$ incidence structures respectfully,
and $\bfY$, that is nothing more than a bipartite graph $(A,B)$,
can be viewed as a classical code of length $|A|$ with parity check matrix $H$
being the $B\times A$ incidence matrix. With these conventions, 
\begin{itemize}
\item the quantum code described by the $2$-complex $\bfX$ has length $n(\bfX)=|X_1|$,
dimension $k(\bfX)=\dim H_1(\bfX)$, $X$-distance $d_X(\bfX)=S_1(\bfX)$ and
$Z$-distance $d_Z(\bfX)=S^1(\bfX)$.
\item the classical code described by the $1$-complex $\bfY$ has dimension
$k(\bfY)=\dim H_1(\bfY)$, and minimum distance $d(\bfY)=S_1(\bfY)$.
\end{itemize}

\begin{definition}\label{def:product}
For complexes $\bfX=(X_0,X_1,X_2)$ and $\bfY=(A,B)$,
we define the complex $\X$ to be the co-complex of the product $X^*\otimes Y^*$.
Specifically, we have
\begin{align*}
\X_0&= (X_0\times A)\cup (X_1\times B)\\
\X_1&= (X_1\times A)\cup (X_2\times B)\\
\X_2&= X_2\times A
\end{align*}
with boundary map $\partial_1^\X$ defined as 
$\partial^\X_1=(\partial_1^X\otimes\Id_{A})+(\Id_{X_1}\otimes\partial_1^Y)
$
over $\f_2^{X_1\times A}$ and as
$(\partial_2^X\otimes\Id_B)$ over $\f_2^{X_2\times B}$
and with boundary map 
$\partial_2^\X=(\partial_2^X\otimes \Id_A) +
(\Id_{X_2}\otimes\partial_1^Y)$.
\end{definition}

\begin{theorem}
\label{thm:product}
For a $2$-dimensional complex $\bfX=(X_0,X_1,X_2)$ and a $1$-dimensional complex
$\bfY = (A,B)$ such that $H_2(\bfY)=0$, the associated complex $\X$ 
has:
\begin{align*}
\dim H_1(\X)&=\dim H_1(\bfX)\dim H_1(\bfY)\\
S_1(\X)     &=S_1(\bfX)S_1(\bfY)\\
S^1(\X)     &=S^1(\bfX).
\end{align*}
\end{theorem}

\paragraph{Comments.}
\begin{enumerate}
\item 
The condition $H_2(\bfY)=0$ in Theorem~\ref{thm:product} is equivalent to saying
that the $(A,B)$-incidence matrix $H$ that defines the associated classical code
has no redundant rows.
In coding terms, Theorem~\ref{thm:product} says that the quantum code associated
to $\X$ has dimensional equal to the product of the dimensions associated to the
quantum code defined by $\bfX$ and to the classical code defined by $\bfY$. It also
says that the $X$-distance of the resulting quantum code is equal to the
product $d_X(\bfX)d(\bfY)$ of the $X$-distance of the original quantum code and
the distance of the classical code.
\item
The theorem enables us to transform a quantum code with unbalanced $X$ and $Z$
distances into one with balanced distances. The theorem generalises the distance
balancing
construction of Hastings \cite{Has17} which corresponds to the special case when
$\bfY$ is simplicial and describes a graph which is isomorphic to a path:
equivalently this is the special case when the classical code associated to $\bfY$ is a
repetition code of dimension $1$. 

To balance the distances of the original quantum LDPC code, supposing 
$d_Z\gg d_X$, one takes a classical LDPC code of minimum distance
$d\approx d_Z/d_X$. Taking an asymptotically good classical LDPC code
(which is known to exist and can be constructed \cite{SS}), one 
obtains a quantum
LDPC code $\X$ of minimum distance $\approx d_Z$ at the cost
multiplying the original quantum code distance by approximately a constant times
$d$.

In Hasting's original construction, the dimension of the new quantum code is the
same as that of the original quantum code. 
Theorem~\ref{thm:product} has the advantage of boosting the new code dimension
by multiplying it with a quantity commensurable with $d$.
\item
Taking the tensor product of the co-complexes of $\bfX$ and $\bfY$ may seem
unwieldy, and since an abstract chain complex can indifferently be read from left
to right or from right to left without changing its nature, one could be tempted
to use a definition that avoids co-complexes altogether. However, if the
component co-complex $\bfX$ is {\em simplicial}, which will be the case in our
applications,
we really must use its co-complex in the product. This contravariant behaviour
was already apparent in \cite{TZ}, where a quantum code is constructed from two
classical codes: this construction consists of tensoring the complex describing
one code with the co-complex describing the other. The construction of
Theorem~\ref{thm:product} can therefore be also viewed as a generalisation of
\cite{TZ}, where one of the two classical codes is replaced by a quantum code.
\end{enumerate}

\begin{proof}[Proof of Theorem~\ref{thm:product}]
The statement on the dimension of $H_1(\X)$ is the straightforward application
of K\"unneth's formula, and the fact that $H_1(\bfY)$ is the only homology of
$\bfY$ that is non-zero, since we have supposed $H_2(\bfY)=0$.

We first compute the $1$-systole of $\X$. Consider a $1$-chain of $\X$, which is
a vector $\x$ in $\f_2^{X_1\times A}\times\f_2^{X_2\times B}$. Suppose first its
restriction to its $X_2\times B$ coordinates is zero. Then, $\x$ is a $1$-cycle,
if and only if 
its $X_1\times A$ component $\x'$ belongs to
$\ker\partial_1^{\bfX}\otimes\ker\partial_1^{\bfY}$. 
In other words, $\x$ is a $1$-cycle if and only if $\x'$,
viewed as a $X_1\times A$-array, has $1$-cycle of $\bfX$ in each of its columns and
$1$-cycle of $\bfY$ in each of its rows. If we take
$\x'=\z_{\bfX}\otimes\z_{\bfY}$, where $\z_{\bfX}$ is a non-trivial $1$-cycle of
$\bfX$ and $\z_{\bfY}$ is a $1$-cycle of $\bfY$, we obtain a $1$-cycle of $\X$ that
cannot be a $1$-boundary, because elements of $\Im\partial_2^{\X}$ always put
$2$-boundaries of $\bfX$ on the columns of the $X_1\times A$-array.
If $\z_{\bfX}$ and $\z_{\bfY}$ have weights that are $1$-systoles in 
their respective $H_1$ groups, we obtain a non-trivial $1$-cycle of $\X$ of
weight $S_1(\bfX)S_1(\bfY)$. Hence $S_1(\X)\leq S_1(\bfX)S_1(\bfY)$.

By the same argument, by taking elements of a basis of $H_1(\bfX)$ for
$\z_{\bfX}$ and elements of a basis of $H_1(\bfY)$ for $\z_{\bfY}$,
we can create $\dim H_1(\bfX)\dim H_1(\bfY)$ elements of $H_1(\X)$ that have
zero $X_2\times B$ component and are equal to $\z_{\bfX}\otimes\z_{\bfY}$ on
their $X_1\times A$ component. By K\"unneth's formula, these $1$-cycles generate
the whole of $H_1(\X)$ and viewed as a $X_1\times A$ array, any linear
combination of these basis elements has at least $S_1(\bfY)$ columns that are
non-trivial $1$-cycles of $\bfX$. Adding a $2$-boundary of $\X$ to the whole
vector will only add a $2$-boundary
of $\bfX$ to any of these columns, therefore the weight of these columns is
always at least $S_1(\bfX)$. This proves $S_1(\X)\geq S_1(\bfX)S_1(\bfY)$ and
hence $S_1(\X) = S_1(\bfX)S_1(\bfY)$.

It remains to compute $S^1(\X)$. Consider a cochain $\x\in \f_2^{X_1\times
A}\times\f_2^{X_2\times B}$ and suppose it has zero component in
$\f_2^{X_2\times B}$. We again view its $\f_2^{X_1\times A}$ component as an $X_1\times A$ array.
We check easily that $\x$ is a $1$-cocycle of $\X$ if and only if every column of
the $X_1\times A$-array is a $1$-cocycle of $\bfX$. Let $H$ be the $B\times A$
incidence matrix describing $\bfY$ so that for every $b\in B$, the row of $H$ indexed
by $b$ has support equal to $\delta_1^{\bfY}(b)$. Let $A=A'\cup A''$ be a
partition of $A$ such that the $B\times A'$ submatrix of $H$ is square and
non-singular. We have $|A''|=|A|-|B|=\dim H^1(\bfY)$.

Now consider a basis of $H^1(\bfX)$ and 
all arrays $\z_{\bfX}\otimes a$ consisting of
a single non-zero column in a position $a\in A''$ and equal to an element of this basis.
We note that there are $|A''|\dim H^1(\bfX)$ such arrays.
Let $\x'$ be any linear combination of these arrays. adding any image by
$\delta_0^{\bfX}\otimes\Id_A$ of a vector of $\f_2^{X_0\times A}$ only adds $1$-coboundaries
of $\bfX$ in each column, and adding any $\Id_{X_1}\otimes\delta_0^{\bfY}$
image of a non-zero element of $\f_2^{X_1\times B}$ adds a non-zero element of
the $B\times A'$ subarray. Therefore $\x'$ cannot be equal to a $1$-coboundary
of $\X$, and we have exhibited a basis of $H^1(\X)$ by K\"unneth's formula.

Finally, let $a''\in A''$ be the index of a column of $\x'$ as above which hosts a non-trivial
$1$-cocycle of $\bfX$. Its weight is at least $S^1(\bfX)$. When we add to it the
$\delta_0^{\bfX}\otimes\Id_A$ image of an element of $\f_2^{X_0\times A}$ we
only add $1$-coboundaries of $\bfX$ to this column, so that its weight stays at
least $S^1(\bfX)$. And when we add to it the $\Id_{X_1}\otimes\delta_0^{\bfY}$ image
of any non-zero element of $\f_2^{X_1\times B}$, we have that for any
$x_1\in X_1$, whenever a non-zero coordinate $x_1\times a$ is removed,it must be
compensated by some $x_1\times a'$ coordinate, for some $a'\in A'$.
We have just proved that $S^1(\X)\geq S^1(\bfX)$.
Finally, a cochain with zero $X_2\times B$ component and whose
$X_1\times A$ component is an array
consisting of a single non-zero column $a''\in A''$, hosting
a non-trivial $1$-cocycle of $\bfX$ of weight equal to $S^1(\bfX)$, must be a
non-trivial $1$-cocycle of $\X$. Therefore $S^1(\X)\leq S^1(\bfX)$ and we have
proved $S^1(\X)=S^1(\bfX)$.
\end{proof}

\section{Ramanujan complexes} \label{sec-Ram}

The purpose of this section is to prove the following two Theorems, which give
a construction of families of bounded degree simplicial complexes of dimension $d=2,3$, that give rise to non-trivial quantum codes with parameters $d_X\cdot d_Z \geq \Omega (n(\log n)^{d-1})$. 
Both Theorems rely heavily on the work of \cite{KKL}.
We note that both the constructions as well as the constants mentioned in the following Theorems can be given explicitly.

\begin{theorem} \label{thm-Ram-main-1}
There exists an infinite family of $2$-dimensional bounded degree complexes $\{\bfX_i\}_i$, $|\bfX_i|\rightarrow \infty$, with non-trivial first cohomology
\begin{equation}
H^1(\bfX_i) \ne 0,\nonumber
\end{equation}
which satisfy the following systolic and cosystolic lower bounds
\begin{equation}
S^1(\bfX_i) \geq c |\bfX_i|, \qquad \qquad S_1(\bfX_i) \geq c'
 \log|\bfX_i|,\nonumber
\end{equation}
where $c$ and $c'$ are absolute positive constants.
\end{theorem}

\begin{theorem} \label{thm-Ram-main-2}
There exists an infinite family of $3$-dimensional bounded degree complexes $\{\bfX_i\}_i$, $|\bfX_i| \rightarrow \infty$, with non-trivial second cohomology
\begin{equation}
H^2(\bfX_i) \ne 0,\nonumber
\end{equation}
which satisfy the following systolic and cosystolic lower bounds
\begin{equation}
S^2(\bfX_i) \geq c |\bfX_i|, \qquad \qquad S_2(\bfX_i) \geq c' 
(\log|\bfX_i|)^2,\nonumber
\end{equation}
where $c$ and $c'$ are absolute positive constants.
\end{theorem}

In this section, $\bfX$ will denote a finite pure simplicial complex, and
$|\bfX|$ will denote its number of vertices $|X_0|$. 
Note that if $\bfX$ is of $Q$-bounded degree, $Q\in \bN$, i.e. each vertex is contained in at most $Q$ faces in $\bfX$, then for any $p=0,1,\ldots,d=\dim(\bfX)$, the number of  $p$-dimensional faces in $\bfX$ is bounded from above by $ \leq Q |\bfX|$ and from below by $|\bfX|$ (because of the purity assumption).
In particular, for a family of bounded degree complexes, $\{\bfX_i\}_i$, $|\bfX_i| \rightarrow \infty$, and for any $p=0,1,\ldots,d=\dim(\bfX)$, the number of $p$-dimensional faces in $\bfX_i$ grows like $|\bfX_i|$, up to some constant depending on the degree.

\subsection{Explicit Ramanujan complexes}

We begin by stating a main result of \cite{LSV2} which gives an explicit
construction of Ramanujan complexes. We note that the actual Ramanujan property
will not concern us for the purposes of this paper (for the interested reader we
recommend the survey \cite{L1}).

Throughout this section we shall use the following notations: 
Let $d \in \bN$ and let $q$ be an odd prime power.
Let $\bF_q$ be the finite field of $q$ elements and $F=\bF_q((t))$ the field of Laurent series over $\bF_q$.
Let $PGL_{d+1}(F) = GL_{d+1}(F)/\mbox{center}$, be the group of invertible $(d+1)\times (d+1)$ matrices over $F$ divided by the scalar matrices.

The Bruhat-Tits building associated to $PGL_{d+1}(F)$, denoted $\mB = \mB_d(F)$,
is a $d$-dimen\-sional pure simplicial complex which is contractible and admits a transitive action of the group $PGL_{d+1}(F)$, which one should think of as an higher dimensional analogue of the infinite regular trees (for more details see \cite{L1}).
Since finite regular graphs are (from a topological standpoint) finite quotients of the infinite regular tree, the finite simplicial complexes we will present will be finite quotients of Bruhat-Tits buildings.

For a pair $(G,\Sigma)$ of a group (not necessarily finite) $G$ and a finite set of generators $\Sigma \subset G$, define its associated Cayley complex, denoted $Cay(G,\Sigma)$, to be the clique complex of the Cayley graph associated to $(G,\Sigma)$  (which is usually also denoted by $Cay(G,\Sigma)$). Recall that the clique complex of a graph is the simplicial complex whose faces are the cliques of the given graph. Similarly, if $H\leq G$, define its associated Schreier complex, denoted $Sch(G/H ,\Sigma)$, to be the clique complex of the Schreier graph associated to $(G,H,\Sigma)$.

\begin{theorem} \cite[Theorem~1.1]{LSV2} \label{thm-Ram-LSV}
For every $d\geq 2$ and $q$ a prime power, there is an (explicit) infinite arithmetic subgroup $\Gamma_0 \leq PGL_{d+1}(F)$, and an (explicit) finite set of generators $\Sigma \subset \Gamma_0$, such that $\Gamma_0$ acts simply transitive on the vertices of the Bruhat-Tits building $\mB$ of $PGL_{d+1}(F)$ and $\Sigma$ is the subset that moves a vertex to all of its neighbours. 
Hence, the Bruhat-Tits building is isomorphic to the Cayley complex of $\Gamma_0$ with respect to the set of generators $\Sigma$,
\begin{equation}
\mB \cong Cay(\Gamma_0 , \Sigma).\nonumber
\end{equation}
In particular, for any congruence subgroup $\Gamma \lhd \Gamma_0$ the quotient of the Bruhat-Tits building $\mB$ by $\Gamma$ is isomorphic to the Schreier complex of the the finite cosets space $\Gamma_0/\Gamma$ with respect to the set of generators $\Sigma$, i.e.
\begin{equation}
\bfX_{\Gamma} = \Gamma \backslash \mB \cong Sch(\Gamma_0/\Gamma , \Sigma).\nonumber
\end{equation} 
\end{theorem}

The heart of this construction of \cite{LSV2} is the arithmetic group  $\Gamma_0$ constructed by Cartwright and Steger.
We will give the details of this construction in the appendix.

\subsection{Non-vanishing of cohomology}

Throughout this section we shall work with the following notations:
Let $d\geq 2$, $q$ an odd prime power and let $\Gamma_0$ be the arithmetic group of Theorem \ref{thm-Ram-LSV}.

The fact that $\Gamma_0$ is an arithmetic group means that it is realized as a group of matrices defined over the ring $R = \bF_q[t,t^{-1},(1+t)^{-1}]$.
For any ideal $0 \ne I \lhd R$, the modulo $I$ map from $R$ onto $R/I$ induces a group homomorphism from $\Gamma_0$ to a matrix group over $R/I$.
Therefore $\Gamma_0$ admits infinitely many finite index normal subgroups, called principle congruence subgroups
\begin{equation}
\Gamma_I := \{g\in \Gamma_0 \; | \; g \equiv 1 (\,\mbox{mod}\,I\,) \} \lhd
\Gamma_0 , \qquad \qquad I\lhd R.\nonumber
\end{equation}
Consequently we get that the following is an infinite family of finite $d$-dimensional complexes of bounded degree (which depends only on $d$ and $q$)
\begin{equation}
\mX := \{\bfX_{\Gamma} = \Gamma \backslash \mB\;:\; \Gamma \mbox{ is a finite index subgroup of } \Gamma_0\}.\nonumber
\end{equation}

To be more concrete, if $p(t)\in \bF_q[t]$ is an irreducible polynomial of degree $e \geq 2$, and $I=p(t)R$ the ideal generated by it, note that $R/I \cong \bF_q[t]/(p)\cong \bF_{q^e}$, then the above mentioned group homomorphism is $ (\,\mbox{mod}\,I\,)\;:\; \Gamma_0 \rightarrow PGL_{d+1}(\bF_{q^e})$, whose image contains the subgroup $PSL_{d+1}(\bF_{q^e})$ (the last claim is non-trivial, for more details we refer the reader to the paper \cite{LSV2}), and therefore
\begin{equation}
|\bfX_{\Gamma_I}| = |\mbox{image}(\,\mbox{mod}\,I\,)| \sim |PSL_{d+1}(\bF_{q^e})| \sim q^{e\cdot (d^2+2d)}.\nonumber
\end{equation}

The question is, can we find a subfamily of the above family $\mX' \subset \mX$ such that for any $\bfX_{\Gamma} \in \mX'$ the first and/or second cohomology does not vanish?
Another question, which will be important for the systolic lower bound, is whether this subfamily can be comprised of subgroups which are contained in a principal congruence subgroup and such that the index of these groups are not too large?
Note that we have a lot of freedom in how we pick $q$ and the subgroups $\Gamma$.
The following two results shows that for certain choices of $\Gamma$ we get complexes with non-trivial first and second cohomology. 

\begin{proposition} \label{pro-Ram-L}
Let $\Gamma_I \leq \Gamma_0$ be a principal congruence subgroup of level $0 \ne I \lhd R$. 
Then there exists an ideal $J \lhd R$, $J \subset I$ satisfying $[I:J] \leq [R:I]$, 
as well as a subgroup $\Gamma_J \subset \Gamma \subset  \Gamma_I$ such that $\Gamma$ has a non-trivial abelian quotient of $2$-power order.
\end{proposition}

\begin{proof}
Let $f \in \bF_q[t]$ be an irreducible polynomial $f \not \in I$ of degree $\deg(f) \leq [R:I]$.
Take the ideal $J$ to be the product ideal of $I$ and $(f)$ and note that  $[I:J]\leq \deg(f)$.
Let $S_2$ be the $2$-Sylow subgroup of the finite group $\Gamma_I / \Gamma_J$. 
Then take $\Gamma$ to be the preimage of the modulo $J$ map of $S_2$.
\end{proof}

\begin{proposition} \cite[Propositions~3.5,3.6]{KKL} \label{pro-Ram-KKL}
Let $\Gamma \leq \Gamma_0$ be a finite index subgroup which has a non-trivial abelian quotient of $2$-power order.
Then 
\begin{equation}
H^1(\bfX_{\Gamma}) \ne 0,\nonumber
\end{equation}
and if $d\geq 3$ then also
\begin{equation}
H^2(\bfX_{\Gamma}) \ne 0.\nonumber
\end{equation}
\end{proposition}

We are in a position to answer the above questions.
Start with a sequence of ideals of $R$, $\ldots \subset I_{i+1} \subset I_i \subset \ldots \subset I_1 \subset R $. 
Apply Proposition \ref{pro-Ram-L} on the principal congruence subgroups $\Gamma_{I_i} \leq \Gamma_0$ and get $\Gamma_{J_i} \leq \Gamma_i \leq \Gamma_{I_i}$ admitting non-trivial abelian quotient of $2$-power order. 
Then by Proposition \ref{pro-Ram-KKL} the associated complexes $\bfX_i = \bfX_{\Gamma_i}$ have non-trivial first and/or second cohomology.
The above construction is simply a repeated iteration of the following Theorem. 

\begin{theorem} \label{thm-Ram-KKL}
Let $\bfX' = \bfX_{\Gamma_I}$ be the quotient of the Bruhat-Tits building by a principal congruence subgroup.
Then there is a finite cover $\bfX = \bfX_{\Gamma}$, i.e. $\Gamma \leq \Gamma_I$, which satisfies the following two properties:
\begin{itemize}
\item First, $H^1(\bfX) \ne 0$ as well as $H^2(\bfX) \ne 0$ if $d \geq 3$.
\item Secondly, $|\bfX| \leq |\bfX'|^2$, i.e. $\log|\bfX| \leq 2 \log|\bfX'|$.
\end{itemize}
\end{theorem}

\begin{proof}
Follows immediately from Propositions \ref{pro-Ram-L} and \ref{pro-Ram-KKL}. 
\end{proof}

\subsection{Cosystolic and systolic lower bounds}

Here bound from below the cosystoles and systoles of the LSV complexes mentioned above.
Our main Theorem require lower bounds on the $1$ and $2$ dimensional cosystoles and systoles, but we shall prove such results in higher generality, namely for any dimension.

First, let us give the lower bound on the cosystoles, which is in fact linear in the size of the complex.
This result does not require any assumption on the subgroup $\Gamma$, only on the $q$, which should be large enough compared  to the dimension of the complex.

\begin{theorem} (\cite{KKL} for $d=2,3$ \cite{EK} for $d\geq 3$) \label{thm-Ram-EK}
For any $d\geq 2$ there exists $q_d >0$ and $c_d >0$ such that, in the notations of Theorem \ref{thm-Ram-LSV}, for any  $q \geq q_d$ and for any finite index subgroup $\Gamma \leq \Gamma_0$, the quotient $\bfX_{\Gamma}$ satisfy the following cosystolic bound for any $k<d$, 
\begin{equation}
S^k(\bfX_{\Gamma}) \geq c_d \cdot |\bfX_{\Gamma}|.\nonumber
\end{equation}
\end{theorem}

Secondly, we want to give a lower bound on the systoles.  
This lower bound will be only polylogarithmic in the size of the complex, where the exponent of the polylogarithm is the dimension of the systole.
This result require an assumption on the injectivity radius of the $\bfX_{\Gamma}$ to be proportional to the logarithm of the size of the complex.
Let us recall the definition of the injectivity radius.

Let $\bfX$ be a simplicial complex, $\tilde{\bfX}$ its universal cover, $\pi(\bfX)\leq Aut(\tilde{\bfX})$ its fundamental group and $P_{\bfX} \,:\, \tilde{\bfX} \rightarrow \bfX$ its projection map.
For example, $\bfX= \bfX_{\Gamma} = \Gamma \backslash \mB$, $\tilde{\bfX_{\Gamma}} = \mB$, $\pi(\bfX_{\Gamma}) = \Gamma$ and $P_{\bfX_{\Gamma}}(x) = \Gamma x$.
Then the injectivity radius of $\bfX$, denoted $r(\bfX)$, is defined to be the minimal $r\in \bN$ such that $P_{\bfX}$ is injective on balls of radius $r$ in $\tilde{\bfX}$.
This is the same as saying that for any $x\in \tilde{\bfX}$ and any $1 \ne \gamma \in \Gamma$, the distance between $x$ and $\gamma.x$ is at least $\frac{r}{2}$ (up to an error of $\pm 1$).

The following result which follows from the work of \cite{LM} says that the
quotients of $\mB$ by principal congruence subgroups, have a logarithmic lower
bound on their injectivity radius.
(See \cite{LM} for a result in the other direction showing that this is in fact optimal.)

\begin{proposition}\cite[Proposition~3.3]{LM} \label{pro-Ram-LM}
For any $d\geq 2$ there exists $c_d >0$ such that for any principal congruence subgroup $\Gamma_I \leq \Gamma_0$, 
the injectivity radius of $\bfX_I = \bfX_{\Gamma_I}$ is bounded from below by 
\begin{equation}
r(\bfX_I) \geq \frac{1}{2d^2(d+2)}  \log_q|\bfX_I| - 1.\nonumber
\end{equation}
\end{proposition}

In the case of non-principal congruence subgroups which are contained in principal congruence subgroups and with a polynomial bound on their index, we have the following simple Lemma which essentially tell us that Proposition \ref{pro-Ram-LM} holds more generally. 
Note that this is exactly the case we need for the groups coming from Theorem \ref{thm-Ram-KKL}.

\begin{lemma} \label{lem-Ram-index}
Let $\bfX'$ be a principal congruence subgroup and $\bfX$ a finite index cover of $\bfX'$ such that $|\bfX| \leq |\bfX'|^2$.
Then 
\begin{equation}
r(\bfX) \geq \frac{1}{4d^2(d+2)}  \log_q|\bfX| - 1.\nonumber
\end{equation} 
\end{lemma}

\begin{proof}
By definition of the injectivity radius, it is obviously non-decreasing when taking covers, combined with Proposition \ref{pro-Ram-LM},
\begin{equation}
r(\bfX) \geq r(\bfX') \geq  \frac{1}{2d^2(d+2)}  \log_q|\bfX'| - 1 \geq \frac{1}{4d^2(d+2)} \log_q|\bfX| - 1. \nonumber
\end{equation}
\end{proof}

Now we are in a position to prove a systolic lower bound in the simple case of $1$-dimensional systoles.

\begin{theorem} \label{thm-Ram-sys-1}
For any $d\geq 2$ there exists $c_d >0$ such that for any principal congruence subgroup $\Gamma_I \leq \Gamma_0$, 
the quotient $\bfX_I = \bfX_{\Gamma_I}$ satisfy the following  systolic bound in dimension $1$ 
\begin{equation}
S_1(\bfX_I) \geq \frac{1}{2d^2(d+2)} \log_q |\bfX_I| - 2.\nonumber
\end{equation}
By Lemma \ref{lem-Ram-index} we get a similar result for finite covers $\bfX$ of $\bfX_I$ such that $|\bfX| \leq |\bfX_I|^2$.
\end{theorem}

\begin{proof}
By Proposition \ref{pro-Ram-LM} it suffices to prove that $S_1(\bfX_I) \geq r(\bfX_I)$.
Let $z \in Z_1(\bfX) \setminus B_1(\bfX)$ be a minimal non-trivial $1$-cocycle, considered as a collection of edges $z \subset X_1$.
If $|z| < r(\bfX_I)$ then there exists a ball of radius $r(\bfX_I)$ in $\bfX_I$ which contains $z$.
But such a ball is isometric to a ball in the universal covering building, which implies that the ball is contractible, therefore it admits no non-trivial cycles,which implies $z\in  B_1(\bfX_I)$, in contradiction to our assumption, hence $|z| \geq r(\bfX_I)$ as needed.
\end{proof}

Next we prove the more involved case of higher dimensional systoles.

\begin{theorem} \label{thm-Ram-sys-2}
For any $d\geq 2$ there exists $c_d >0$ such that for any principal congruence subgroup $\Gamma_I \leq \Gamma_0$, 
the quotient $\bfX_I = \bfX_{\Gamma_I}$ satisfies the following  systolic bound in dimension $k<d$,
\begin{equation}
S_k(\bfX_I) \geq c_d  ( \log_q |\bfX_I| )^k.\nonumber
\end{equation}
By Lemma \ref{lem-Ram-index} we get a similar result for finite covers $\bfX$ of $\bfX_I$ such that $|\bfX| \leq |\bfX_I|^2$.
\end{theorem}

\begin{proof}
First note that by Proposition \ref{pro-Ram-LM}, the injectivity radius of $\bfX=\bfX_I$ is $R=r(\bfX) \geq c'_d \cdot \log |\bfX|$. 
Let $z$ be a minimal non-trivial $k$-cocycle of $X$, and let $\sigma$ be some $k$-face contained in the support of $z$. 
Then the ball $B=B(\sigma,R) \subset \bfX$ of radius $R$ around $\sigma$ looks like a ball in the covering Bruhat-Tits building, and the intersection $z_B = z \cap B$ of the minimal non-trivial cocycle with this ball is a minimal non-trivial cocycle relative to the boundary of the ball, i.e. $\partial(z_B) \subset \partial B$.
Note that by assumption that $z$ is a minimal non-trivial cocycle, $z_B$ is the minimal $k$-chain inside $B$ with boundary $\partial(z_B) $ (if there is a smaller one $c$, we get that $z' = z - z_B +c$ contradicts $z$ being a minimal non-trivial cocycle).

Next we use the fact that if $\mA$ is an apartment containing $\sigma$ inside the building $\mB$, then there is a simplicial retraction map $\rho$ from $\mB$ to $\mA$, which preserves the distance to $\sigma$ (this actually determines this map uniquely) and furthermore this map does not increase distances \cite{AB}.
Consider the ball $B$ as a subset of $\mB$, and let $B^{\mA}$ be the intersection of $B$ with the apartment $\mA$, and note that $\mA$ is a tessellation of the Euclidean $d$-dimensional space with geometric $d$-simplexes, hence $B_{\mA}$ is a Euclidean ball of dimension $d$, hence it has approximately $R^d$ maximal faces.
Similarly, if $B^k_{\mA} = B^k_{\mA}(\sigma,R) \subset B_{\mA}$ is a $k$-dimensional Euclidean ball that passes through the center $\sigma$ of $B$, then $B^k_{\mA}$ has volume approximately $R^k$.

Now, denote $z_{B,\mA} = \rho(z_B) \subset B_{\mA}$. Note that $\partial(z_a) \subset \partial B_{\mA}$ and moreover $z_{B,\mA} $ is still the minimal $k$-chain inside $B_{\mA}$ with boundary $\partial(z_{B,\mA}) $ by the same argument from before. We are left with proving that the size of $z_{B,\mA}$ (which bounds from below the size of $z_{B}$ which bounds from below the size of $z$) is bounded from below by the size of $B^k_{\mA}$, this will gives us the claim.

The above lower bound follows from the monotonicity Theorem for minimal surfaces in a Euclidean space, see for instance \cite[Theorem~21]{GL}, which says that the ratio 
\begin{equation}
\frac{|z_{B,\mA} \bigcap B^k_{\mA}(\sigma,r)|}{|B^k_{\mA}(\sigma,r)|}\nonumber
\end{equation}
is non-decreasing for $0<r\leq R$. Since $\sigma \in z$, hence the ratio for $r=1$ is a non-zero constant $c''$, then we get that 
\begin{equation}
|z_{B,\mA} \bigcap B^k_{\mA}(\sigma,R)| \geq c'' \cdot | B^k_{\mA}(\sigma,R)| \sim R^k,\nonumber
\end{equation}
as needed, which completes the proof.
\end{proof}

\subsection{Proof of Theorems \ref{thm-Ram-main-1} and \ref{thm-Ram-main-2} }

Now we combine all the results in this section to prove Theorems \ref{thm-Ram-main-1} and \ref{thm-Ram-main-2}.
Their proofs proceed as follows:
First, all the complexes we shall consider are of the form of Theorems \ref{thm-Ram-LSV}, i.e. quotients of the Bruhat-Tits building by finite index subgroups of the Cartwright-Steger group $\Gamma_0$, and to be even more precise the subgroups will be congruence subgroups.
Second, we use Theorem \ref{thm-Ram-KKL} to construct a sequence of quotients of the Bruhat-Tits building all of which have non-trivial first or second cohomology.
Third, by Theorem \ref{thm-Ram-EK} we get a linear lower bound on the cosystoles.
Finally, by Theorem \ref{thm-Ram-sys-1} we get a logarithmic lower bound on the $1$-systoles, and by Theorem \ref{thm-Ram-sys-2} we get a quadratic logarithmic lower bound on the $2$-systoles. 

\section{Decoding product complexes} \label{sec-dec-prod}
We now focus on decoding the tensor product of a $2$-dimensional simplicial
complex $\bfX$
with a $1$-dimensional complex $\bfY$.
We shall write $V=X_0$ for the set of vertices of $\bfX$, $E=X_1$ for its edge set,
$T=X_2$ for its triangle set. Similarly, we write $A=Y_1$ and $B=Y_0$.
We can think of $(V,E,T)$ and $(A,B)$ as incidence structures, in particular
$(A,B)$ is defined by a bipartite graph between $A$ and $B$. We recall that what
we require of $(A,B)$ is that~: 
\begin{itemize}
\item the coboundary map $\delta_{B\to A}~:
\f_2^{B}\to\f_2^{A}$ has zero kernel, in other words the $|B|\times |A|$
incidence matrix $H_{AB}$ of the bipartite graph $(A,B)$ has rank $|B|$.
\item the matrix $H_{AB}$ is the parity-check matrix of a classical LDPC code
$C_{AB}$ of minimum distance $d_{AB}\geq c|A|$ for some constant $c$.
\end{itemize}

We will also require that the classical code $C_{AB}$ comes with a decoding
algorithm that is guaranteed to correct all errors of weight up to a fraction of
its minimum distance. The expander codes of \cite{SS} are known to achieve this. 

The quantum code $\cQ=\cQ(\X)$ associated to the product complex 
of Definition~\ref{def:product} has the factor graph
representation depicted on Figure~\ref{fig:product}.

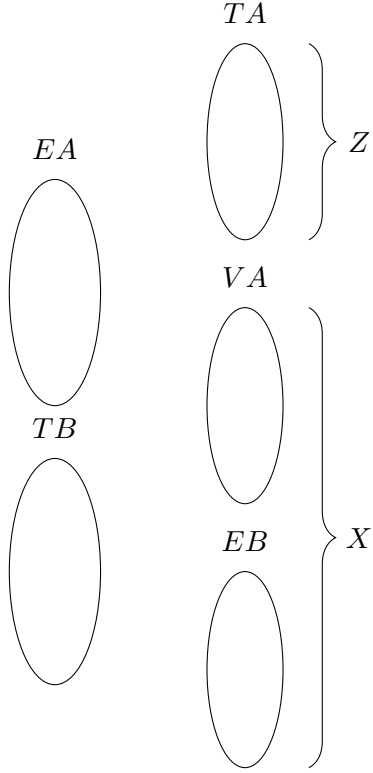
\begin{figure}[h]
\begin{center}
\begin{tikzpicture}
\draw (0,0) ellipse (0.6cm and 1.5cm);
\node (EA) at (0,1.9) {$EA$};
\draw (0,-3.7) ellipse (0.6cm and 1.5cm);
\node (TB) at (0,-1.8) {$TB$};
\draw (2.5,2) ellipse (0.5cm and 1.3cm);
\node (TA) at (2.5,3.7) {$TA$};
\draw [decorate,decoration={brace,amplitude=10pt,mirror,raise=4pt},yshift=0pt]
(3.2,0.7) -- (3.2,3.3) node [black,midway,xshift=0.8cm] {$Z$};
\draw (2.5,-1.5) ellipse (0.5cm and 1.3cm);
\node (VA) at (2.5,0.2) {$VA$};
\draw (2.5,-5) ellipse (0.5cm and 1.3cm);
\node (EB) at (2.5,-3.3) {$EB$};
\draw [decorate,decoration={brace,amplitude=10pt,mirror,raise=4pt},yshift=0pt]
(3.2,-6.3) -- (3.2,-0.2) node [black,midway,xshift=0.8cm] {$X$};
\end{tikzpicture}
\end{center}
\caption{The factor graph structure of the quantum code $\cQ$ associated to the
homological product of $(V,E,T)$ with $(A,B)$. The code length is
$N=|E||A|+|T||B|$.}
\label{fig:product}
\end{figure}

The quantum code now has coordinate (variable) set $\N= (E\times A) \cup (T\times B)$ that we abbreviate to $\N = EA \cup TB$. We have two syndrome functions
\begin{align*}
\sigma_Z~: \f_2^{\N} &\rightarrow \f_2^{TA}\\
\sigma_X~: \f_2^{\N} &\rightarrow \f_2^{{V\! A}\, \cup EB} = \f_2^{V\! A}\oplus\f_2^{EB}.
\end{align*}

Following the conventions of Section~\ref{sssec:quantum} we view $\sigma_X$ as a
boundary map and $\sigma_Z$ as a coboundary map. This point of view is helpful
since these maps 
inherit properties of the coboundary and boundary maps of the
simplicial complex $\bfX$.

We adopt the following convention:  we refer to the $1$-cycles of the
$2$-complex $\bfX$,
i.e. the cycles of the underlying graph $(V,E)$, simply as {\em cycles.} 
We shall call the elements
of $C_X=\ker\sigma_X$ in the product complex $\X$ as {\em Cycles} (capital C). 
By {\em trivial Cycle} (or Boundary) we shall mean an element of $\Im{\sigma_Z^*}$.
Similarly, we will talk about cocycles in $\bfX$ ($1$-cocycles) and coCycles in the
product complex $\X$, i.e. elements of $C_Z=\ker\sigma_Z$. To identify easily the $1$-boundary
maps $\partial_1$ in their respectives complexes $\bfX$ and $\bfY$ we will write
$\partial_{E\to V}$ and $\partial_{A\to B}$. Finally we will typically denote a
vector (in $\f_2^A$, $\f_2^{EA}$, etc.) by bold letters, but also will regularly
identify vectors with their supports to lighten notation. For example an element $a\in A$ 
will regularly also denote the vector of $\f_2^A$ whose support is $\{a\}$.
Hopefully this abuse will not introduce confusion.

For the quantum code $\cQ(\X)$ we examine separately the cases of decoding
$X$-errors and $Z$-errors since the situation is quite asymetrical.
In both cases the goal is to correct a constant fraction of the minimum distance
of $\cQ(\X)$.

\subsection{Decoding $X$-errors}
\begin{theorem}
\label{thm:Zerrors}
Suppose the classical LDPC code $C_{AB}$ comes with a polynomial-time decoding algorithm that
corrects any pattern of less than $\alpha|A|$ errors. Then
there is a polynomial time algorithm that given $\sigma_X(\x)$ for
$\x\in\f_2^\N$ of weight smaller than $\alpha |A|S_1(X)/2$, returns $\x + \u$
where $\u\in\Im(\sigma_Z^*)$. 
\end{theorem}

Recall from Section~\ref{sssec:quantum}, that the algorithm of
Theorem~\ref{thm:Zerrors} returns precisely a solution to the decoding problem.

We now describe the decoding strategy.

\paragraph{EA representation.} Let $\x\in\f_2^{\N}$ be a arbitrary chain. We
claim that there is a trivial Cycle $\v\in\Im{\sigma_X^*}$ such that $\x+\v$
has all its non-zero coordinates in $EA$. This is because the map
$\partial_{A\to B}$ is surjective, meaning that for every $b\in B$ there is
a set $A_b\subset A$, such that $b=\partial_{A\to B}(A_b)$. 
So for
every coordinate
$tb\in TB$ that is in the support of $\x$ we can add (this is not an algorithmic
procedure, just an existence result)
the $\sigma_Z^*$-image of the set
$t\times A_b$. 
We shall call such a sum $\x+\v$ an {\em EA-representation} of $\x$ 
(it is not unique). Decoding from $\sigma_X(\x)$ will consist of looking for 
an EA-representation of $\x$.

\paragraph{First decoding step: decoding from the VA part of 
$\sigma_X(\x)$.}
We focus on the VA-component of the $X$-syndrome and notice that it is the disjoint union,
for  $a\in A$, of the syndromes of all the $Ea$-components of $\x$. 
In other words, if we write:
\[
\x = \sum_{a\in A}\x_a\otimes a + \x_{TB}
\]
where $\x_a\in\f_2^E$ and $\x_{TB}$ has its support inside $TB$,
then:
\[
\sigma_X(\x)_{\restVA} = \sum_{a\in A}\partial_{E\to V}(\x_a)\otimes a.
\]
The boundaries $\partial_{E\to V}(\x_a)\otimes a$ are in $Va$ and disjoint, and
the first decoding step consists simply of decoding
from every $V_a$-component of the $X$-syndrome
$\partial_{E\to V}(\x_a)$ to obtain a candidate for $\x_a$. 
This decoding procedure occurs, as just mentioned, inside the $(V,E)$ graph,
so we may apply the polynomial-time complete decoding procedure mentioned in
Section~\ref{sssec:classical}. This returns the smallest weight vector $\x_a'$
to $\x_a$ such that $\x_a+\x_a'$ is a cycle. Whenever $\x_a$ has smaller weight than
half the $1$-systole $S_1(\bfX)/2$ we have, since $|\x_a'|\leq |\x_a|$, 
that $\x_a+\x_a'$ must be a trivial cycle.

For the purpose of clarity, we first describe the rest of the decoding procedure
in a simple case which will help to follow the general situation.
\paragraph{Case when $(A,B)$ is a path.}
Let us suppose the graph $(A,B)$ describes the edge-vertex incidence structure
of a path, as in \eqref{eq:path}. The associated classical code $C_{AB}$ is the
repetition code, i.e. the code of dimension $1$ generated by the all-one vector.

\begin{equation}
\label{eq:path}
\begin{tikzpicture}
\node (a1) at (0,0) {$a_1$};
\node (a2) at (2,0) {$a_2$};
\node (a3) at (4,0) {$a_3$};
\node (an-1) at (7,0) {$a_{m-1}$};
\node (an) at (9,0) {$a_{m}$};
\draw (a1) -- node[above] {$b_1$} (a2) -- node[above] {$b_2$} (a3);
\draw[dashed] (a3) -- (an-1);
\draw (an-1) -- node[above] {$b_{m-1}$} (an);
\end{tikzpicture}
\end{equation}

We have in this case $d_X(\cQ)=mS_1(\bfX)$ where $m=|A|$. 
If we use the simple majority decoder for the repetition code, we can correct
any pattern of errors of weight $<|A|/2$ and the hypothesis on the
weight of the error vector $\x$ in Theorem~\ref{thm:Zerrors} translates into
$|\x|<mS_1(\bfX)/4$.
Under this hypothesis, we have that the number of $a\in A$ such that $\x_a$ is closer
to a non-trivial cycle than to $0$, is less than $|A|/2$. 

\paragraph{Situation after the first decoding step.} The first decoding step yields
a vector $\y=\sum_{a\in A}\y_a\otimes a$ such that
\[
\sigma_X(\y))_{\restVA} = \sigma_X(\x)_{\restVA}.
\]
The vector $\x+\y$ is therefore such that each of its $a$-components $\x_a+\y_a$,
for all $a\in A$,
is a cycle, and from the discussion just above we have that 
a strict minority of them are non-trivial.

\paragraph{Second (and final) decoding step.} 
The decoder computes $\sigma_X(\x)+\sigma_X(\y)=\sigma_X(\x+\y)$ and tries to recover
an EA-representation of $\z=\x+\y$. Without loss of generality we suppose that
$\z$ is equal to one of its EA-representations. Switching from $\z$ to one of its
EA-representations changes its weight but does not change the nature of its $a$-components
$\z_a$ which remain either trivial cycles or non-trivial cycles: in particular the
fact that a minority of $a$-components of $\z$ are non-trivial is unchanged and this is
the only feature used in the coming decoding argument.

We now have to deal with an $X$-syndrome whose VA-component is zero, and we are
left with an EB-component from which to decode.
We have
\[
\s = \sigma_X(\z) = \sigma_X(\z)_{\restEB} = \sum_{a\in A}\sigma_X(\z_a\otimes a). 
\]
We may decompose $\s$ into $b$-components, $b\in B$,
\[
\s = \sum_{b\in B} \s_b\otimes b
\]
and using the path structure \eqref{eq:path}
of $A-B$ we have, for $i=1,\ldots m-1$,
\[
\s_i = \s_{b_i} = \z_{a_i} + \z_{a_{i+1}}.
\]
Rewrite $\z_i=\z_{a_i}$ to lighten notation.
Given $\s$ and starting from $a_1$, the decoder may therefore construct the vector 
$\z'\in\f_2^{EA}$, $\z'=(\z_1',\ldots ,\z_n')$, $\z_i'\in \f_2^E$,
setting
\begin{eqnarray*}
\z_1' &= &0\\
\z_2' &= &\z_1+\z_2\\
\z_3' &= &(\z_2+\z_3) + \z_2'\\
&\vdots&\\
\z_i' &=& (\z_{i-1}+\z_i) + \z_{i-1}'\\
&\vdots&
\end{eqnarray*}
so that $\sigma_X(\z')=\sigma_X(\z)$. We see that we have
\[
\z'=(\z_1,\z_2,\ldots,\z_n) + (\z_1,\z_1,\ldots ,\z_1).
\]
Similarly, the decoder can construct the alternative candidate vectors 
for $\z'$,
\begin{equation}
\label{eq:zi}
\z'=(\z_1,\z_2,\ldots,\z_n) + (\z_i,\z_i,\ldots ,\z_i)
\end{equation}
for all values $i=1,\ldots ,n$.

We see that $\z'$ is equal to $\z$ up to addition of a trivial Cycle if and only if
$\z_i$ is a trivial cycle in the complex $\bfX$. When $\z_i$ is non-trivial,
then every trivial component of $\z$ becomes non-trivial in $\z'$. Therefore, the decoder
may differentiate between the two cases $\z_i$ trivial/non-trivial in \eqref{eq:zi} 
by computing, for every component of $\z'$, whether it is a trivial cycle or not.
Note that this test is obviously polynomial-time since it just involves testing
whether a vector $z_i$ belongs to a well-identified vector space or not and is
achieved with elementary linear algebra.
When it finds a majority of trivial components, it knows it is in the case "$\z_i$ trivial",
and outputs $\z'$. We have that $\y+\z'$ is equal to the original error vector $\x$
up to addition of a trivial Cycle.
This concludes the decoding algorithm in the case when $(A,B)$ is a path.

\paragraph{Case of general bipartite graphs $(A,B)$.}

As before, we try to recover an EA-representation of $\z=\x+\y$. Without loss
of generality we assume $\z$ is one of those EA-representations. The chain
$\z$ therefore has a syndrome $\sigma_X(\z)$ with zero VA-component. 

We first recover an arbitrary chain $\z'$ from the syndrome $\sigma_X(\z)$ by
picking any solution to the linear system. We obtain therefore $\z'$ such that
$\sigma_X(\z')=\sigma_X(\z)$. Now we define the subcode $\C$ of $C_X$ with zero
TB-component, i.e. the code in $\f_2^{EA}$ made up of
those vectors with zero $\sigma_X$ syndrome. This is exactly the tensor code
\[
\C = \ker\partial_{E\rightarrow V}\otimes\ker\partial_{A\rightarrow
B} = \ker\partial_{E\rightarrow V}\otimes C_{AB} =Z_1(\bfX)\otimes Z_1(\bfY)
\]
so that we can write $\z'=\z+\c$, with $\c\in\C$. We also have that both $\z$ and $\z'$ live
naturally in the tensor product space
\[
\ker\partial_{E\rightarrow V}\otimes\f_2^{A}=Z_1(\bfX)\otimes C_1(\bfY)
\]
(as opposed to the initial $\f_2^E\otimes\f_2^A$). We can therefore use a basis
$\Z$ of cycles of $\ker\partial_{E\rightarrow V}$ and express $\z'$ in the basis
of elementary tensors $z\otimes a$, $z\in \Z$, $a\in A$.
Finally, we also have that the subcode $\C'$ of $\C$
\begin{equation}
\label{eq:C'}
\Im\partial_{T\rightarrow E}\otimes\ker\partial_{A\rightarrow B}=B_1(\bfX)\otimes
Z_1(\bfY)
\end{equation}
is exactly the set of EA-representations of $\sigma_Z^*(\f_2^{TA})$.
Therefore, we only need to recover $\z$ up to an element of $\C'$.

So we use a cycle basis $\Z$ of the form $\Z=\Z_0\cup \Z_1$ where $\Z_0$ is
a basis of the boundary space $\Im\partial_{T\rightarrow E}$. We now identify
the chain $\z'$ as an element of $\f_2^\Z\otimes\f_2^A$~:
in concrete terms, this means we have identified $\z'$
with a $|\Z|\times |A|$ array that we obtain by elementary linear algebra.
This array is partitioned into the union of a $|\Z_0|\times |A|$ subarray
and a  $|\Z_1|\times |A|$ subarray corresponding to the spaces 
$\f_2^{\Z_0}\otimes\f_2^A$ and $\f_2^{\Z_1}\otimes\f_2^A$.
Now suppose first that at the first decoding step, the cycle code decoder
that decodes every $Ea$ component has made no error, meaning it recovers for
every $a$
the original $Ea$-component $\x_a$ of the error up to a trivial cycle.
This translates into $\z=\x+\y$ being entirely inside
$\f_2^{\Z_0}\otimes\f_2^A$, and having a zero component inside the $|\Z_1|\times |A|$ 
subarray. In this case the $\f_2^{\Z_1}\otimes\f_2^A$ component of $\z'$,
viewed as a $|\Z_1|\times |A|$ array, has rows that are all codewords of
$C_{AB}$, and to obtain $\z$ up to an element of $\C'$, we simply need to remove
these codewords and put the $\f_2^{\Z_1}\otimes\f_2^A$ component at zero.
Of course, we can't expect that there will be no errors during the first
decoding step: but our hypothesis  on the weight of the error vector, namely
$|\x|< \alpha |A|S_1(\bfX)/2$, implies that the cycle code decoder will
 add a non-trivial cycle to $\x_a$, for less than $\alpha |A|$ values of
$a$. This translates into the number of non-zero columns of $\z$ in its 
$|\Z_1|\times |A|$ subarray component being less than $\alpha |A|$. This means in
particular that every one of the rows of the subarray has weight less than $\alpha |A|$.
Now on each of these rows $\z'$ is equal to $\z$ plus a codeword of $C_{AB}$, that
we need to remove to recover $\z$ from $\z'$ up to a vector of $\C'$.
Identifying and removing this codeword is always possible by applying the
decoding procedure for $C_{AB}$ that corrects up to $\alpha|A|$ errors.
Once we have $\z$ up to a vector of $\C'$ we add it to $\y$ to obtain an
equivalent version of the original error vector $\x$ and we are done.

\subsection{Decoding $Z$-errors}\label{sssec:Xerrors}
Let us say that the $2$-dimensional complex $(V,E,T)$ {\em corrects $w$ $Z$-errors} if
there is a polynomial-time algorithm that:
given the $2$-coboundary $\delta_{E\to T}(\error)$ of a cochain $\error\in \f_2^E$ of
Hamming weight at most $w$, outputs $\error + \c$ where $\c\in\Im\delta_{V\to
E}$ is a $1$-coboundary. Note that this means exactly that the quantum code
associated to the $2$-complex in the sense of Section~\ref{sssec:quantum}
corrects $w$ $Z$-errors, hence the terminology.

Turning once more to the quantum code $\cQ$  associated to the product of complexes
$\bfX$ and $\bfY$ associated to $(V,E,T)$ and $(A,B)$ we have the result:

\begin{theorem}
\label{thm:Xerrors}
Suppose the $2$-dimensional complex $(V,E,T)$ corrects $w$ $Z$-errors. Then
there exists a polynomial-time algorithm that given $\sigma_Z(\x)$ for
$\x\in\f_2^{\N}$ of weight at most $w$ returns $\x+\u$ where
$\u\in\Im\sigma_X^*$. 
\end{theorem}

In other words Theorem~\ref{thm:Xerrors} says that if the quantum code
associated to the component $2$-complex $(V,E,T)$ can correct $w$
$Z$-errors, then so can the product quantum code $\cQ(\X)$.

Again, it is natural to look at the decomposition of the error vector:
\[
\x = \x_{EA} + \x_{TB}
\]
with $\x_{EA}=\sum_{a\in A}\x_a\otimes a$. If we suppose that the error vector
$\x$ is entirely supported in $EA$, then we have 
\[
\sigma_Z(\x) = \sum_{a\in A}\sigma_Z(\x)_{\restTa}
\]
with
\[
\sigma_Z(\x)_{\restTa}=\delta_{E\to T}(\x_a)\otimes a
\]
Decoding would then consist of
recovering in parallel $\x_a$ or an equivalent cochain from
$\delta(\x_a)$ for every $a$. Obviously if $|\x|\leq t$ then
$|\x_a|\leq w$ for every $a$ and we can apply the decoding algorithm for the
$2$-dimensional simplicial complex $\bfX$.

However, when $\x_{TB}\neq 0$, this straightforward approach
breaks down because every $\sigma_Z(\x)_{\restTa}$ need not be
a copy of a $2$-coboundary anymore. To bypass this problem we look for a special
equivalent form of $\x$.

Recall that
the map $\delta_{B\to A}~: \f_2^B\rightarrow \f_2^A$ has zero kernel. This implies that
there exists $A'\subset A$, $|A'|=|B|$, such that the restricted linear map
\[
\delta_{B\to A'}~: \f_2^{B}\rightarrow \f_2^{A'},
\]
defined by restricting the support of every vector of $\Im\delta$ to $A'$,
is one-to-one. Define $A''=A\setminus A'$. Recall that two $Z$-error vectors are
said to be equivalent if they differ by a vector of $\sigma_X^*$.
\begin{lemma}
\label{lem:reduced}
{\bf (Reduced cochain).}
Let $\x\in\f_2^{\N}$. There exists an equivalent vector
$\x'$ such that the EA-component of $\x'$ is entirely supported by
$EA''$, in other words $\x_a'=0$ for every $a\in A'$.
Furthermore, we have that the weight of
every $Ea$-component of $\x'$, for $a\in A''$, is upper bounded as:
$|\x'_{Ea}|\leq |\x_{EA}|$. In particular, every $Ea$-component of $\x'$ is 
upper bounded by the total weight of $\x$.
\end{lemma}

\begin{proof}
For every $a\in A'$, there exists a subset $B_a\subset B$, such that
$\delta_{B\rightarrow A'}(B_a)=\{a\}$.
Given the decomposition of $\x$
\[
\x = \sum_{a\in A}\x_a\otimes a + \x_{TB}
\]
we construct $\x'$ as:
\[
\x'=\x + \sigma_X^*\left(\sum_{a\in A'}\sum_{b\in B_a}\x_a\otimes b\right).
\]
which clearly deletes all $EA'$ coordinates of $\x$.
We also see that for any $a\in A''$, an $e$ coordinate is added to the support
of $x_a$, only if there is at least one $a\in A'$ such that $x_a$ contains $e$
in its support. This implies that the weight of $\x_a'$ cannot exceed the total
weight of $\x$.
\end{proof}

When switching from $\x$ to its reduced form, we may obtain
a cochain with larger weight but the $EA$ component of the
reduced chain has weight at most $w$.
Since the weight of
the $EA$ component will turn out to be the only relevant one
for the decoding argument, we may therefore assume that the
error vector $\x$ is already in the reduced form given by
Lemma~\ref{lem:reduced}.

\paragraph{The decoding algorithm}
We first recover the $TB$ component $\x_{TB}$ of $\x=\x_{EA}+\x_{TB}$.
The syndrome map $\sigma_X$, when applied to the $TB$ component, is
one-to-one when restricting its image to $TA'$, since it is equal to
$Id\otimes\delta_{B\to A'}$ and we have chosen $A'$ such that $\delta_{B\to A'}$
is one-to-one. Therefore we can deduce the $TB$ component $\x_{TB}$ of $\x$ from
the $TA'$ component of the $\sigma_Z(\x)$. To find the $EA$ component of $\x$, 
we only need to decode from $\sigma_Z(\x)+\sigma_Z(\x_{TB})$.
This puts us back in the situation when the $Z$-error vector has
no $TB$-component, and we can just apply the $w$-error-correcting algorithm
to every $Ta$ component of the syndrome. This concludes the proof of
Theorem~\ref{thm:Xerrors}.

\section{Coboundary decoding of a $2$-dimensional Ramanujan complex} \label{sec-dec-Ram}
\subsection{The decoding algorithm}\label{ssec:coboundary-decoding}
The decoding algorithm of Section~\ref{sssec:Xerrors} rested upon the existence
of a decoding algorithm for the quantum code associated to the
$2$-dimensional simplicial complex $\bfX=(V,E,T)$. In this section we show the
existence of such an algorithm for $2$-dimensional Ramanujan complexes. For now
we just need to remember that the local properties of a Ramanujan complex are
described by a local parameter $q$ which is a prime power and describes the
local degrees. The edge degree is $q+1$, meaning that every edge is incident to
$q+1$ triangles, and the vertex to edge decreed is $Q=2q^2+2q+2$, meaning that
every vertex has $Q$ neighbours.
We recall that we need an algorithm that takes as input the coboundary
$\delta_1(\error)\in\f_2^T$ for some vector $\error\in\f_2^E$, with the
an upper bound on the error weight $|\error|\leq t$, and such that the algorithm
outputs an equivalent error vector $\error + \c$ where $\c\in\Im\delta_0$.

It will be convenient to think of $\error$ as a $1$-cochain of the complex,
equivalently a subset of the edge set $E$.

The algorithm that we will exhibit will be {\em local.} 
\paragraph{Description of the decoding algorithm:}
Suppose $\delta_1(\error)\neq 0$ otherwise there is nothing to do (output $0$).

Step 0: the algorithm sets $\ff_0=\delta_1(\error)$.

Step $k>0$:
the algorithm looks for a vertex 
$v_k$ and a cochain $\y_k$, entirely inside the edge-neighbourhood
$\delta_0(v_k)$
of $v_k$, such that 
\[
|\ff_{k-1}+\delta_1(\y_k)| < |\ff_{k-1}|.
\]
Upon finding such a $v_k$, it sets $\ff_k=\ff_{k-1}+\delta^1(\y_k)$.
If $\ff_k=0$, the algorithm stops and outputs
\[
\error' = \y_1+\y_2+\cdots +\y_k.
\]
If $\ff_k\neq 0$ it proceeds to step $k+1$.

In words, the algorithm looks for a vertex that admits a small set $\y$ of incident
edges such that the coboundary $\delta_1(\y)$, when added to $\delta_1(\error)$,
yields a smaller weight then $\delta_1(\error)$. It then repeats the operation
with $\delta_1(\error)+\delta_1(\y)=\delta_1(\error +\y)$ and iterates until it
reaches a zero coboundary.

For any $1$-cochain $\error$, let us denote by $\error_v$ the cochain supported
by the edges of $\error$ that are incident to $v$.

To show that the algorithm always converges and gives a right answer, we shall
prove the following theorem:
\begin{theorem}
\label{thm:local}
There exists constants $\gamma>0$ and $q_0$, such that whenever $q>q_0$, for every cochain $\error\in
C^1(X)=\f_2^E$ with weight $|\error|\leq\gamma |E|$ and with
$|\delta_1(\error)|>0$, there exists a vertex $v\in
V$, such that $|\delta_1(\error +\error_v)|<|\delta_1(\error)|$.
\end{theorem}

Proving that the decoding algorithm converges involves a second ingredient. 
Let us say that a cochain $\error$ is {\em minimal}, if it is of smallest weight
in its class modulo $B^1(\bfX)$. Let us say that it is {\em locally minimal},
if for any $v\in V$, $|\error + \delta_0(v)|\geq |\error|$. 
Since every edge is incident to $(q+1)$ triangles,  
the weight of the coboundary $\delta_1(\error)$ 
of any cochain $\error$ is at most $(q+1)|\error|$.

\begin{proposition}\label{prop:1/3}
For $\gamma$ as in Theorem~\ref{thm:local}, whenever a locally minimal cochain
$\error$ has weight $|\error|\leq \gamma |E|$, then
$|\delta_1(\error)|\geq\frac 13(q+1)|\error|$.
\end{proposition}

A form of Proposition~\ref{prop:1/3} (with a looser constant) is in \cite{KKL}. Theorem~\ref{thm:local}
and Proposition~\ref{prop:1/3} will be a consequence of Theorem~\ref{thm:gamma}
below.

Theorem~\ref{thm:local} and Proposition~\ref{prop:1/3} imply:

\begin{theorem}
\label{thm:localdecoding}
For $\gamma$ as in Theorem~\ref{thm:local}, assuming the Ramanujan complex
is sufficiently large, namely $|E|\geq \frac{(q+1)Q}{4\gamma}$, we have that
any error vector $\error\in\f_2^E$ of weight $|\error|\leq \frac 13\gamma
|E|$ is always correctly decoded by the decoding algorithm.
\end{theorem}

\begin{proof}
Set $\error_0=\error$ and for $k\geq 1$,  $\error_k=\error + \y_1+\cdots \y_k$
where $\y_k$ is as specified by the decoding algorithm at the $k$th step.
We remark that we can suppose $|\y_k|\leq Q/2$ for every $k$,
since $\y_k$ is defined inside the neighbourhood of a vertex $v_k$,
and $\y_k$ and $\y_k+\delta_0(v_k)$ (its complement in the neighbourhood of
$v_k$) have the same $\delta_1$ coboundary. Therefore, for the first
$k_0=(q+1)Q/6\leq \frac 23\gamma|E|$ steps of the algorithm, we are guaranteed to
have $|\error_k|\leq \gamma |E|$.
At the end of these $k_0$
steps, if the algorithm has not terminated, and since the coboundary
of $\error_k$ decreases by at least $1$ at every step, 
we have $|\delta_1(\error_{k_0})|\leq|\delta_1(\error)|-(q+1)Q/6$.
Since $|\delta_1(\error)|\leq (q+1)|\error|\leq(q+1)\frac 13\gamma |E|$, we get
\[
|\delta_1(\error_{k_0})|\leq \frac 13(q+1)(\gamma |E|-Q/2).
\]
Since we know that $|\error_{k_0}|\leq \gamma |E|$,
Proposition~\ref{prop:1/3} now implies that $\error_{k_0}$ is
equivalent to a minimal cochain of weight not more than
$\gamma |E|-Q/2$. Since $\y_{k_0+1}$ has weight at most $Q/2$
we get that $\error_{k_0+1}$ is equivalent to a cochain of weight
not more than $\gamma |E|$, and applying again the same argument,
we have that it is also equivalent to a cochain of weight at most
$\gamma |E|-Q/2$. Iterating, we get that $\error_k$ is always equivalent to a
cochain of weight at most $\gamma |E|$. In particular Theorem~\ref{thm:local}
always applies, and the algorithm must terminate 
with some cochain
$\error'=\y_1+\cdots \y_k$ that has the same coboundary as $\error$.
Since a minimal cochain equivalent to $\error +\error'=\error_k$ must 
have weight not more than $\gamma |E|$, this minimal cochain
must be zero by Proposition~\ref{prop:1/3}. So $\error'$ is a correct solution to the decoding problem.
\end{proof}

\paragraph{Estimation of the constant $\gamma$.} Theorem~\ref{thm:gamma} below
will show that $\gamma$ in Theorems~\ref{thm:local} and \ref{thm:localdecoding}
can be taken to be arbitrarily close to $1/192$. We recall from
\cite{KKL} that the $1$-cosystole of the Ramanujan complex can,
for large $q$, be bounded from below by a quantity arbitrarily close to $|E|/4$.
We have $4\gamma/3=1/144$,
in other words, the decoding algorithm is shown to decode errors of weight up to
a $1/144$ fraction of the designed coboundary distance.

\subsection{Analysis and proof of Theorem~\ref{thm:local}}
We first proceed to translate the statement of Theorem~\ref{thm:local} into 
what it means in terms of the triangles of the complex that are incident to
edges of $\error$. Without loss of
generality, we may suppose that $\error$ is locally minimal. (In fact we also could suppose it to be minimal, but only local minimality will be needed).

Let $T_1$ (resp. $T_2,T_3$) denote the set of triangles that have exactly one edge
(resp. two edges, three edges) in $\error$. For a vertex $v$ define $T_1(v, good)$ to be the
set of triangles of $T_1$ incident to $v$ and containing an edge of $\error$ incident
to $v$ and $T_1(v, neutral)$ to be the set of triangles of $T_1$ incident to $v$ and
containing no edge of $\alpha$ incident to $v$. Let $T_2(v,bad)$ be the set of triangles of $T_2$ incident to $v$ containing exactly
one edge of $\alpha$ incident to $v$, and let $T_2(v,neutral)$ be the set of
triangles of $T_2$ incident to $v$, containing two edges of $\alpha$ incident to
$v$. 

We remark that $\delta_1(\error)$ is the $2$-cochain consisting of the union of
the triangles of $T_1$ and $T_3$. We also remark that when add
$\delta_1(\error_v)$ to $\delta_1(\error)$, the set $T_1(v, good)$ disappears
from $\delta_1(\error)$ and the set $T_2(v, bad)$ is added to
$\delta_1(\error)$. The other triangles incident to $v$ do not intervene in the
operation. This is illustrated in Figures~\ref{fig:T1} and~\ref{fig:T2} where
the edges of $\error$ are in blue.

Summarising, we have
$|\delta_1(\error)+\delta_1(\error_v)|<|\delta_1(\error)|$ if and only if
\[
|T_1(v,good)|>|T_2(v,bad)|.
\]
Now, summing over all $v\in V$ we observe that:
\begin{equation}
\label{eq:T1}
\sum_{v\in V}|T_1(v,good)| = 2|T_1|
\end{equation}
and
\begin{equation}
\label{eq:T2}
\sum_{v\in V}|T_2(v,bad)| = 2|T_2|.
\end{equation}
Therefore, whenever $|T_1|>|T_2|$ there must exist a vertex such that
$|T_1(v,good)|>|T_2(v,bad)|$. So to prove Theorem~\ref{thm:local}, we only need
to prove that for any sufficiently small cochain $\error$, we have
$|T_1|>|T_2|$.

From now on we set $t_i=|T_i|$ for $i=1,2,3$.

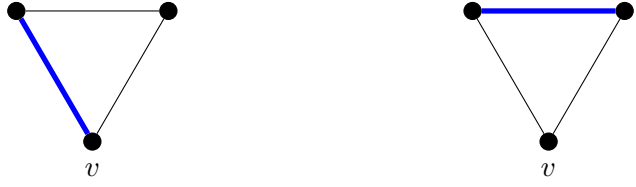
\begin{figure}
\begin{center}
\begin{tikzpicture}
\tikzstyle{every node}=[fill, circle, inner sep=0.08cm, draw]
\node[label=below:{$v$}] (a) at (0,0) {};
\node (b) at (60:2cm) {};
\node (c) at (120:2cm) {};

\draw (a) -- (b) -- (c);
\draw[color=blue, line width=0.7mm] (a) -- (c); 

\node[label=below:{$v$}] (a') at (6,0) {};
\node[xshift=6cm] (b') at (60:2cm) {};
\node[xshift=6cm] (c') at (120:2cm) {};

\draw[color=blue, line width=0.7mm]  (b') -- (c') {};
\draw (b') -- (a') -- (c'); 

\end{tikzpicture}
\end{center}
\caption{the left triangle is in $T_1(v,good)$: it {\em is} in the coboundary of
$\error$, and when $\error_v$ is flipped, it {\em disappears} from the coboundary.
The right triangle is in $T_1(v,neutral)$: it {\em is} in the coboundary of
$\error$
and when $\error_v$ is flipped it {\em stays} in the coboundary of $\error$.}
\label{fig:T1}
\end{figure}

\begin{figure}
\begin{center}
\begin{tikzpicture}
\tikzstyle{every node}=[fill, circle, inner sep=0.08cm, draw]
\node[label=below:{$v$}] (a) at (0,0) {};
\node (b) at (60:2cm) {};
\node (c) at (120:2cm) {};

\draw (a) -- (b);
\draw[color=blue, line width=0.7mm] (a) -- (c) -- (b);

\node[label=below:{$v$}] (a') at (6,0) {};
\node[xshift=6cm] (b') at (60:2cm) {};
\node[xshift=6cm] (c') at (120:2cm) {};

\draw (b') -- (c');
\draw[color=blue, line width=0.7mm] (b') -- (a') -- (c');
\end{tikzpicture}
\end{center}
\caption{the left triangle is in $T_2(v,bad)$: it {\em is not} in the coboundary of
$\alpha$, and when $\error_v$ is flipped, it {\em is} in the coboundary.
The right triangle is in $T_2(v,neutral)$: it {\em is not} in the coboundary of
$\error$
and when $\error_v$ is flipped it {\em stays out} of the coboundary of $\error$.}
\label{fig:T2}
\end{figure}
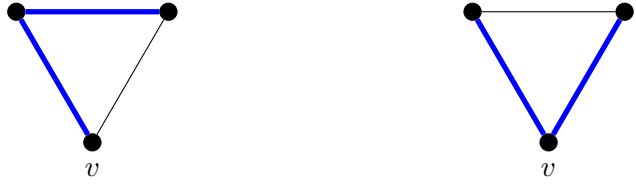

We recall the relevant properties of the $2$-dimensional Ramanujan complex
$\bfX=(V,E,T)$. Its underlying graph $(V,E)$ is a $Q$-regular graph with
$Q=2(q^2+q+1)$ for $q$ a prime power. It is an expander graph with the second
largest eigenvalue of its adjacency matrix at most $6q$ (\cite[Prop.~2.1.]{LSV1}). Furthermore,
the link $L(v)$ of any vertex $v$ is isomorphic to the vertex-edge incidence
graph of a projective plane of order $q$, which is well-known to be
a $(q+1)$-regular graph with eigenvalues $\pm (q+1)$ and $\pm\sqrt{q}$.
Recall that
for any vertex $v\in V$, the {\em link} of $v$ is defined as the graph $L(v)$ over the
$Q$ neighbours of $v$, with any two neighbours $u,w$ of $v$ being connected in
$L(v)$ whenever $u,v,w$ is a triangle of $T$.

We recall the classical relation between expansion and spectra of graphs.
Let $G=(V,E)$ be a finite connected graph, $A$ its adjacency matrix and $\Delta$
its Laplacian, i.e., $\Delta:L^2(X) \rightarrow L^2(X)$ defined by
$\Delta(f)(v)=deg(v)f(v)-\sum_{y \thicksim v}f(y)$ where the sum is over the
neighbours of $v$ and $\thicksim$ stands for adjacency in $G$. If $G$ is
$k$-regular then $\Delta=kI-A$. For $W_1,W_2\subset V$, let $E(W_1,W_2)$ denote
the set of edges for vertices of $W_1$ to vertices of $W_2$, and let $\bar{W}$
denote the complement of $W$ in $V$.
We have the following result 
that goes back to Alon and Milman (see e.g. \cite{HLW}):

\begin{proposition}~\label{prop-cheeger} Let $\lambda=\lambda_1(X)$ be the smallest positive eigenvalue of $\Delta$.
\begin{enumerate}
\item For every subset $W \subseteq V$,
$$ |E(W,\bar{W})| \geq \frac{|W||\bar{W}|}{|V|} \lambda_1(X),$$ 

\item If $X$ is $k$-regular then $E(W):=E(W,W)$ satisfies:
$$ 
E(W)=\frac{1}{2}(k|W|-E(W,\bar{W})) \leq \frac{1}{2}(k-\frac{\bar{W}}{|V|}\lambda_1(X))|W|.
$$
\end{enumerate}
\end{proposition}
 
Some of the lemmas below were used in \cite{KKL}. We include them for the sake of
completeness.

\begin{lemma}~\label{lemma-triangles-counting-in-dim-two}
\begin{enumerate}
\item~\label{item-one-lemma-triangles-counting-in-dim-two}
$t_1+2t_2+3t_3=(q+1)|\error|$.
\item~\label{item-three-lemma-triangles-counting-in-dim-two} $\sum_{v \in
V}|E_{L(v)}(\error_v, \overline{\error_v})|=2t_1+2t_2$.
\end{enumerate}
Here we have identified
 $\error_v$, which is the set of edges in $\error$ touching $v$, with the set of
their endpoints in the link $L(v)$.
$E_{L(v)}(\error_v, \overline{\error_v})$ denotes therefore the set of edges from
$\error_v$ to $\overline{\error_v}$ in $L(v)$.
\end{lemma}

\begin{proof} For point \ref{item-one-lemma-triangles-counting-in-dim-two}. 
we recall that every edge lies on $q+1$ triangles and a triangle which
contributes to $t_i$ contains $i$ edges from $\error$. 

For point \ref{item-three-lemma-triangles-counting-in-dim-two}. we observe that
$E_{L(v)}(\error_v, \overline{\error_v})$ counts the triangles of $T_1(v, good)$
and of $T_2(v,bad)$ (see figures~\ref{fig:T1} and \ref{fig:T2}) and apply
\eqref{eq:T1} and \eqref{eq:T2}.
\end{proof}

Fix now $\epsilon$, $0<\epsilon <1$ to be determined later and define:

\begin{definition} A vertex $v$ incident to an edge of $\error$ 
is called {\em thin} with respect to
$\error$ if $|\error_v| < (1-\epsilon)\frac{Q}{2}$ and {\em thick} otherwise
(note that by our local minimality assumption, $|\error_v| \leq \frac{Q}{2}$ for
every $v$). Denote by $R$ the set of thin vertices and by $S$ the set of thick
vertices.
\end{definition}
Let $r=\sum_{v \in R}|\error_v|$ and $s=\sum_{v \in S}|\error_v|$. As every edge
in $\error$ contributes $2$ to $r+s$ we get the following:

\begin{lemma}~\label{lemma-$r+s$}
$r+s = 2|\error|.$
\end{lemma}

\begin{lemma}~\label{lemma-edges-exiting-alpha-v}
\begin{enumerate}
\item~\label{item-one-lemma-edges-exiting-alpha-v} For every $v \in V$, 
$|E_{L(v)}(\error_v, \overline{\error_v})| \geq \frac{1}{2}(q+1-\sqrt{q})|\error_v|$.
\item~\label{item-two-lemma-edges-exiting-alpha-v} If $v$ is thin, then
$|E_{L(v)}(\error_v, \overline{\error_v})| \geq
\frac{1+\epsilon}{2}(q+1-\sqrt{q})|\error_v|$.
\end{enumerate}
\end{lemma}

\begin{proof} As previously recalled, the link
$L(v)$ is a $(q+1)$-regular graph whose eigenvalues are $\pm(q+1)$ and $\pm \sqrt{q}$. Hence,
$\lambda_1(L(v)) = (q+1)-\sqrt{q}$. Part~\ref{item-one-lemma-edges-exiting-alpha-v} now follows from Proposition~\ref{prop-cheeger}, and similarly
part~\ref{item-two-lemma-edges-exiting-alpha-v}.
\end{proof}

We can deduce

\begin{lemma}~\label{lemma-2t-1+2t-2}
$2t_1 + 2t_2 = \sum_{v\in V}E_{L(v)}(\error_v, \overline{\error_v}) \geq
\frac 12(q+1-\sqrt{q})(r+s) + \frac{\epsilon}{2} (q+1-\sqrt{q}) r $.
\end{lemma}

\begin{proof}
\begin{align*}
2t_1 + 2t_2 = \sum_{v \in V}E_{L(v)}(\error_v, \overline{\error_v})
&=    \sum_{v \in R}E_{Y_v}(\error_v, \overline{\error_v}) + \sum_{v \in
S}E_{L(v)}(\error_v, \overline{\error_v}) \\
&\geq  \frac{1+\epsilon}{2}(q+1-\sqrt{q})r + \frac{1}{2}(q+1-\sqrt{q})s \\
&=  \frac{1}{2}(q+1-\sqrt{q})(r+s)+\frac{\epsilon}{2} (q+1-\sqrt{q})r.
\end{align*}
In the first equation we have used
Lemma~\ref{lemma-triangles-counting-in-dim-two}, point
\ref{item-three-lemma-triangles-counting-in-dim-two}. 
The inequality follows from Lemma~\ref{lemma-edges-exiting-alpha-v}.
\end{proof}

\begin{lemma}~\label{lemma-2t-2}
We have:
\[
2t_2 \leq 2\sum_{v \in V}E_{L(v)}(\error_v,\error_v) \leq 
 (q+1)\left(\frac{s+r}{2}-\frac{\epsilon r}{2}\right) +
\frac{\sqrt{q}}{2}(s+r(1+\epsilon)).
\]
\end{lemma}

\begin{proof}

We have
\[
2t_2\leq 2\sum_{v\in v}E_{L(v)}(\error_v,\error_v)
\]
so we can apply Proposition~\ref{prop-cheeger} point 3. and write
\begin{align*}
E_{L(v)}(\error_v,\error_v) &\leq \frac 12 \left(q+1-\frac 12\lambda_1\right)|\error_v|\\
&\leq \frac 12 \left(\frac{q+1}{2}+\frac{\sqrt{q}}{2}\right)|\error_v|
\end{align*}
for thick vertices and
\begin{align*}
E_{L(v)}(\error_v,\error_v) &\leq \frac 12 \left(q+1-\frac
{1+\epsilon}{2}\lambda_1\right)|\error_v|\\
&\leq \frac 12 \left(\frac{q+1}{2}+\frac{\sqrt{q}}{2}-\epsilon\lambda_1\right)|\error_v|
\end{align*}
for thin vertices, where we have used $\lambda_1=q+1-\sqrt{q}$.\\[2mm]
Summing, we get the result.
\end{proof}

Lemmas ~\ref{lemma-2t-1+2t-2} and \ref{lemma-2t-2} give us:
\begin{align*}
2t_1+2t_2 &\geq
(q+1)\left(\frac{s+r(1+\epsilon)}{2}\right) -(s+r(1+\epsilon))\frac{\sqrt{q}}{2}\\
2t_2 &\leq (q+1)\left(\frac{s+r(1-\epsilon)}{2}\right) +(s+r(1+\epsilon))\frac{\sqrt{q}}{2}.
\end{align*}
Substracting the second inequality to the first gives:

\begin{lemma}
\label{lem:2t1}
$2t_1\geq (q+1)r\epsilon -(s+r(1+\epsilon))\sqrt{q}.$
\end{lemma}

\begin{lemma}
\label{lem:t1>t2}
Given any fixed (independent of $q$) $\epsilon$, $\frac 13<\epsilon <1$,
the condition 
\[
r(3\epsilon -1)> s\left(1+O(\frac{1}{\sqrt{q}})\right)
\]
is sufficient to imply $t_1>t_2$. 
\end{lemma}

\begin{proof}
From Lemmas~\ref{lemma-2t-2} and \ref{lem:2t1} we have 
that $t_1>t_2$ is achieved whenever
\begin{align*}
(q+1)r\epsilon &> (q+1)\frac{s+r(1-\epsilon)}{2}+\frac 32
(s+r(1+\epsilon))\sqrt{q}\\
(q+1)r(3\epsilon -1)&> s(q+1) + 3(s+r(1+\epsilon))\sqrt{q}\\
r(3\epsilon -1)\left(1-O(\frac{1}{\sqrt{q}})\right)&>s\left(1+O(\frac{1}{\sqrt{q}})\right)\\
r(3\epsilon -1) &> s\left(1+O(\frac{1}{\sqrt{q}})\right)
\end{align*}
hence the result.
\end{proof}

Up to now we have used only the local structure of $\bfX$, namely the links. Now we
will use the global structure, the fact that its $1$-skeleton is almost a
Ramanujan graph and has second eigenvalue $\leq 6q$.

\begin{lemma}\label{lemma-bound-on-number-of-edges-within-thick-vertices} 
Suppose $|\error|\leq \gamma |E|$ for some constant $\gamma$. Then,
the total number of edges in $\error$ between the thick vertices relative to $\error$
 is bounded as:
$$
|E(S)| \leq |\error|\left(\frac{\gamma}{(1-\epsilon)^2-3\gamma} 
\right)(1+O(\frac 1q)).
$$
\end{lemma}

\begin{proof}
Note that
\[
\sum_{v\in S}|\error_v| \leq |\error|+|E(S)|
\]
since the edges of $E(S)$ are counted twice in this sum.
Furthermore, by definition of $S$, 
$|\error_v|\geq (1-\epsilon)\frac{Q}{2}$ for
$v\in S$, so that $|S|(1-\epsilon)\frac Q2\leq |\error|+|E(S)|$ which we rewrite
as:
\begin{equation}
\label{eq:|S|}
|S|\leq \frac{2}{Q(1-\epsilon)}(|\error|+|E(S)|.
\end{equation}
Since the second largest eigenvalue of the adjacency matrix of $(V,E)$ is
bounded from above by $6q$, we have
$\lambda_1(V,E) \geq Q-6q=2{q^2}-4q+2$. Proposition~\ref{prop-cheeger} implies therefore:
\begin{eqnarray*}
|E(S)| & \leq &
\frac{1}{2}\left(Q-\frac{|\overline{S}|}{|V|}\lambda_1(V,E)\right)|S| \\
     & \leq & \frac{1}{2}\left(Q-\frac{|\overline{S}|}{|V|}(Q-6q)\right)|S| \\
     & =  & \frac{1}{2}\left(Q\left(1-\frac{|\overline{S}|}{|V|}\right)+6q
\frac{|\overline{S}|}{|V|}\right)|S|\\
     & \leq & \frac{1}{2}\left(Q \frac{|S|}{|V|}+6q\right)|S| \\
     & \leq &
\left(\frac{|\error|}{(1-\epsilon)|V|}+\frac{|E(S)|}{|V|(1-\epsilon)}+3q\right)|S|
\end{eqnarray*}
by applying \eqref{eq:|S|}. We rewrite this last inequality as
\[
|E(S)|\left(1-\frac{|S|}{|V|(1-\epsilon)}\right) \leq
\left(\frac{|\error|}{|V|(1-\epsilon)}+3q\right)|S|
\]
We now use the hypothesis $|\error|\leq\gamma|E|=\gamma|V|\frac Q2$
and again invoque \eqref{eq:|S|}
\begin{align}
|E(S)|\left(1-\frac{|S|}{|V|(1-\epsilon)}\right) &\leq
\left(\frac{\gamma}{(1-\epsilon)^2}+\frac{6q}{Q(1-\epsilon)}\right)(|\error|+|E(S)|)\nonumber \\
|E(S)|\left(1-\frac{\gamma}{(1-\epsilon)^2}-\frac{|S|}{|V|(1-\epsilon)}-\frac{3}{q(1-\epsilon)}\right)
&\leq\left(\frac{\gamma}{(1-\epsilon)^2}+\frac{3}{q(1-\epsilon)}\right)|\error|
\label{eq:E(S)}
\end{align}
From \eqref{eq:|S|} we have, since $|E(S)|\leq |\error|$,
\[
|S|\leq\frac{4}{Q(1-\epsilon)}|\error|\leq\frac{4}{Q(1-\epsilon)}\frac{\gamma
|V|Q}{2}=\frac{2\gamma |V|}{1-\epsilon}
\]
and $-\frac{2\gamma}{(1-\epsilon)^2}\leq -\frac{|S|}{|V|(1-\epsilon)}$, which
injected into \eqref{eq:E(S)} gives 
\[
|E(S)|\left(1-\frac{3\gamma}{(1-\epsilon)^2}-\frac{3}{q(1-\epsilon)}\right)
\leq \left(\frac{\gamma}{(1-\epsilon)^2}+\frac{3}{q(1-\epsilon)}\right)|\error|
\]
hence the result after rearranging. 
\end{proof}

We can finally state:
\begin{theorem}
\label{thm:gamma}
Given any $0<\gamma < 1/192$, there exists $q_0$ such that for any $q>q_0$, the
condition $|\error| \leq \gamma|\error|$ implies $t_1>t_2$. Furthermore, we have
$|\delta_1(\error)|\geq \frac
13(q+1)|\error|$.
\end{theorem}

\begin{proof}
We have
\[
s=\sum_{v\in S}|\error_v|\leq |\error|+|E(S)|
\]
since edges of $E(S)$ are counted twice. Applying
Lemma~\ref{lemma-bound-on-number-of-edges-within-thick-vertices} 
we get 
\begin{align}
s &\leq \left(1+\frac{\gamma}{(1-\epsilon)^2-3\gamma}+O(\frac
1q)\right)\nonumber\\
r &\geq \left(1-\frac{\gamma}{(1-\epsilon)^2-3\gamma}-O(\frac
1q)\right)\label{eq:r}
\end{align}
Lemma~\ref{lem:t1>t2} tells us therefore that the condition
\[
(3\epsilon -1)\left(1-\frac{\gamma}{(1-\epsilon)^2-3\gamma}\right)
\geq \left(1+\frac{\gamma}{(1-\epsilon)^2-3\gamma}\right)+O(\frac{1}{\sqrt{q}})
\]
is sufficient to imply $t_1>t_2$.
Rearranging gives the condition
\[
\frac{(3\epsilon-2)(1-\epsilon)^2}{12\epsilon-6}\geq \gamma
+O(\frac{1}{\sqrt{q}}).
\]
the maximum value of the left hand side is $1/192$ which is obtained for
$\epsilon = 3/4$ and gives the required result.

Finally, \eqref{eq:r} together with Lemma~\ref{lem:2t1} give, when
$\epsilon=3/4$,
\[
t_1\geq \frac 38\frac{1-64\gamma}{1-48\gamma}(q+1)-O(\sqrt{q})
\]
and since $(1-48\gamma)/(1-64\gamma)>8/9$ when $\gamma<1/192$ we get $t_1 >
(q+1)/3$ for $q$ large enough. Remembering that $|\delta_1(\error)|\geq t_1$, this
proves the last statement of the Theorem.
\end{proof}

\section{Coboundary decoding of the 2-skeleton of a 3-dimensional Ramanujan
complex}\label{sec:coboundary-decoding-3d}
The decoding algorithm of Section~\ref{sec-dec-Ram} is linear in the code length
$|E|$, but with a large constant which is exponential in $q^2$, where $q$ is the
local degree defining parameter.  This is because the algorithm searches
exhaustively for the required local pattern of edges $\error_v$ inside the edge
neighbourhood of a vertex $v$, which is of size $Q=2(q^2+q+1)$.

We now prove that it is possible to remove this large constant
when we switch from a $2$-dimensional Ramanujan complex to a more complicated
one, namely the $2$-skeleton of a $3$-Ramanujan complex.
A $3$-dimensional simplicial complex $(V,E,T,P)$ comes with an extra layer compared to
the $2$-dimensional one, on top of triangles it has tetrahedra (Pyramids), but
we restrict it to its $2$-skeleton $\bfX=(V,E,T)$ to define a quantum code
$\cQ(\bfX)$ in the same way as before (as opposed to extracting the
complex $(E,T,P)$ that yields a distance record breaking quantum code
through the product complex $\X$, but whose boundaries we don't know how to
decode).

The simplicial complex $\bfX$ has a very different local structure from that of
a $2$-dimensional Ramanujan complex. The link of a vertex $v$ has now the graph
structure of a spherical building, specifically $L(v)$ is isomorphic to the
$3$-partite
graph whose vertices are the points, lines and planes of a $3$-dimensional
projective space over $\f_q$, and where two vertices are connected
if, as geometrical objects, one contains the other.

This modified extra structure allows for a simpler local decoding algorithm. It
proceeds as in Section~\ref{ssec:coboundary-decoding}, with the only difference
that the cochains $\y_k$ are now of weight $1$, i.e. consist of single edges.
Precisely:

\paragraph{Simplified decoding algorithm:}\hfill 

{\em Input:} The coboundary $\ff_0=\delta_1(\error)$ of a cochain $\error$.

{\em Procedure:} for $k\geq 1$, look for an edge $e_k\in E$ such that
$|\delta_1(e_k)+\ff_{k-1}|<|\ff_{k-1}|$. Set $\ff_k=\ff_{k-1}+\delta_1(e_k)$. Repeat
until $\ff_k=0$ and output $\error'=e_1+e_2+\cdots +e_k$.

To show that for any cochain $\error$ of weight less than a constant times
$|E|$, the algorithm always converges a correct solution, i.e.
a cochain $\error'$ equivalent to $\error$, we prove:

\begin{theorem}
\label{thm:local2}
There exists a constant $\gamma$ and an integer $q_0$, such that whenever $q\geq
q_0$, for every cochain $\error\in\f_2^E$ 
with non-zero $1$-coboundary and of weight $\leq \gamma|E|$, there exists $e\in
E$ satisfying $|\delta_1(\error + e)|<|\delta_1(\error)|.$
\end{theorem}

We will also need the following Theorem, a reformulation of Theorem 1.8 of
\cite{KKL}.

\begin{theorem}
\label{thm:3ram}
For any sufficiently large fixed $q$, there exist constants
$\gamma_1,\gamma_2,\epsilon_1,\epsilon_2$, such that:
\begin{enumerate}
\item For any locally minimal$\,^*$ $1$-cochain $\error\in\f_2^E$ of
the $3$-dimensional Ramanujan complex $(V,E,T,P)$, the condition
$|\error|\leq \gamma_1 |E|$ implies $|\delta_1(\error)|\geq
\epsilon_1|\error|$.
\item For any locally minimal $2$-cochain $\ff\in\f_2^T$ of
the $3$-dimensional Ramanujan complex $(V,E,T,P)$, the condition
$|\ff|\leq \gamma_2 |T|$ implies $|\delta_2(\ff)|\geq
\epsilon_2|\ff|$.
\end{enumerate}
\end{theorem}

The first statement of Theorem~\ref{thm:3ram} is similar to 
Proposition~\ref{prop:1/3} in Section~\ref{ssec:coboundary-decoding}. 
There is a subtle difference in that it uses a slightly different notion
of local minimality, which involves replacing the Hamming weight of a $1$-cochain
by a slightly different weight. This is to take into account the fact that the complex
now is irregular in the sense that edges are not all incident to the same
number of triangles. The triangle-to-pyramid degree is always the same however
(and equal to $q+1$), and local minimality of a $2$-cochain means that it is not
possible to decrease its Hamming weight by adding the coboundary of a single
edge. 

Together with Theorem~\ref{thm:local2}, Theorem~\ref{thm:3ram}
shows, by the same argument as in
Section~\ref{ssec:coboundary-decoding} (since the modified weight is bounded from
above by a constant times the ordinary Hamming weight), that:

\begin{theorem}
\label{thm:localdecoding2}
For any sufficiently large fixed $q$, there exists a constant $c$ such that, assuming the Ramanujan complex
is sufficiently large, 
any error vector $\error\in\f_2^E$ of weight $|\error|\leq c|E|$ is always
correctly decoded by the simplified decoding algorithm.
\end{theorem}

The simplified decoding algorithm is of course preferable to that of
Section~\ref{ssec:coboundary-decoding} in terms of complexity. It comes at a
price however, since it involves using  a complex
$\bfX$ with a more involved local structure and larger degrees.
The constant $c$ in Theorem~\ref{thm:localdecoding2} is also
much looser than in Theorem~\ref{thm:localdecoding} and we do not attempt to
estimate it. In particular it is not an absolute constant as in
Theorem~\ref{thm:localdecoding}, but is a decreasing function of $q$.

It remains to prove Theorem~\ref{thm:local2}. Its key is the 
second statement of Theorem~\ref{thm:3ram}. 

\begin{proof}[Proof of Theorem~\ref{thm:local2}]
Let $\error$ be a $1$-cochain with a non-zero $1$-boundary
$\ff=\delta_1(\error)$. Since the complex $\bfX$ is of bounded degrees, by
taking $|\error|$ to be sufficiently small, we can make its coboundary $\ff$ of
smaller weight than $\gamma_2|T|$ in Theorem~\ref{thm:3ram}.
But since $\ff$ is a coboundary we have $\delta_2(\ff)=0$. So
$\ff$ cannot be a locally minimal $2$-cochain, otherwise it would have non-zero
$2$-coboundary by the second statement of Theorem~\ref{thm:3ram}.
That $\ff$ is not locally minimal means that we can decrease the weight of $\ff$
by adding to it the $1$-coboundary of a single edge $e\in E$. This is exactly
the statement of Theorem~\ref{thm:local2}.
\end{proof}

\section*{Appendix}\label{sec-appendix}

In this appendix we provide an explicit construction of the Cartwright-Steger group $\Gamma_0$ from Theorem \ref{thm-Ram-LSV}.

Let $\bF_q$ be the finite field of size $q$, and $\bF_{q^d}$ the field extension of $\bF_q$ of degree $d$.
Let $\phi$ be a generator of the Galois group $Gal(\bF_{q^d}/\bF_q) \cong \bZ/ d\bZ$. 
Fix a basis $\xi_0,\ldots,\xi_{d-1}$ of $\bF_{q^d}$ over $\bF_q$ with $\xi_i = \phi^i(\xi_0)$. 
Denote $R_T = \bF_q[y,\frac{1}{1+y}]$. 
For a given $R_T$-algebra $S$ (i.e. $S$ is given with a ring homomorphism $R_T \rightarrow S$), we define the following $S$-algebra
\begin{equation}
\mA(S) = \bigoplus_{i,j=0}^{d-1}S\xi_i z^j \; : \; z \xi_i = \phi(\xi_i)z \;, \;
z^d = 1+y.\nonumber
\end{equation}
One can see that the center of $\mA(S)$ is $S$, and hence the following is a group scheme for $R_T$-algebras
\begin{equation}
\mG(S) = \mA(S)^*/S^*.\nonumber
\end{equation}
Let $F = \bF_q((y))$ be the local field of Laurent power series.
The algebra $\mA$ splits at $F$, and we get (see \cite[Proposition~3.1]{LSV2})
\begin{equation}
\mG(F)  \cong PGL_d(F).\nonumber
\end{equation}
Let $R = \bF_q[y,\frac{1}{y},\frac{1}{1+y}]$ and define the following set of elements 
\begin{equation}
\Sigma_1 = \{ b_u = 1 - \frac{u}{\phi(u)}\cdot z^{-1} \; : \; u\in \mathbb{F}_{q^d}/\mathbb{F}_q \}  \subset \mathcal{A}(R_T)\nonumber
\end{equation}
By \cite[Proposition~4.1]{LSV2} and the discussions following it, $\Sigma_1$ is in fact a subset of invertible elements in $\mathcal{A}(R)$, and hence one can define the following subgroup
\begin{equation}
\Gamma_0 = \langle \Sigma_1 \rangle \leq \mathcal{G}(R) \leq \mathcal{G}(F)  \cong PGL_d(F).\nonumber
\end{equation}
This group is called the Cartwright-Steger group.
It has the amazing property that it acts simply transitively on the vertices of the Bruhat-Tits building $\mB$ of the group $PGL_d(F)$  (see \cite[Proposition~4.8]{LSV2}).

More explicitly, let $x_0$ be a fixed vertex in the building $\mB$, and let $\tau : \mB \rightarrow \bZ/d\bZ$ be the type function on the building. 
Then the following map is a bijection from the Cartwright-Steger group to the vertices of the building
\begin{equation}
\Gamma_0 \rightarrow \mB(0) \quad : \quad \gamma \mapsto \gamma.x_0 .\nonumber
\end{equation} 
We note that restricting this map to $\Sigma_1$ gives a bijection from $\Sigma_1$ to the set of neighbours $x$ of $x_0$ of type $\tau(x)=1$.
Denote by $\Sigma$ the preimage under this map of the set of all neighbours of $x_0$ in $\mathcal{B}$.
Then we get the following identification of the building with Cayley complex
\begin{equation}
\mB \cong Cay(\Gamma_0,\Sigma).\nonumber
\end{equation}
As claimed in Theorem \ref{thm-Ram-LSV}.

Ramanujan complexes are obtained in \cite{LSV2} by dividing the building modulo the action of congruence subgroups of $\Gamma_0$. 

For any finite index ideal $0 \ne I \lhd R$, define the level $I$ principal congruence subgroup of $\Gamma_0$ to be
\begin{equation}
\Gamma_I = \Gamma_0 \cap \ker(\mG(R) \rightarrow \mG(R/I) ) \lhd \Gamma_0.\nonumber
\end{equation}
A subgroup $\Gamma \leq \Gamma_0$ is called a congruence subgroup of $\Gamma_0$ if it contains s principal congruence subgroup 
\begin{equation}
\exists I\lhd R \qquad : \qquad \Gamma_I \leq \Gamma \leq \Gamma_0.\nonumber
\end{equation}


\end{document}